\documentclass[envcountsame]{llncs}
\usepackage{ifthen}
\newcommand{\techRep}{true} 
\newcommand{\iftechrep}{\ifthenelse{\equal{\techRep}{true}}}
\usepackage{times}
\usepackage{microtype}
\usepackage{mathtools}
\usepackage{amsfonts}
\usepackage{btran}
\usepackage{todonotes}
\usepackage{etoolbox}
\patchcmd{\quote}{\rightmargin}{\leftmargin 0.5em \rightmargin}{}{}

\newenvironment{qtheorem}[1]{%
{\medskip\noindent\bfseries Theorem #1.\hspace{0mm}}
\begin{itshape}%
}{%
\end{itshape}%
}

\newenvironment{qcorollary}[1]{%
{\medskip\noindent\bfseries Corollary #1.\hspace{0mm}}
\begin{itshape}%
}{%
\end{itshape}%
}

\newenvironment{qproposition}[1]{%
{\medskip\noindent\bfseries Proposition #1.\hspace{0mm}}%
\begin{itshape}%
}{%
\end{itshape}%
}

\newcommand{\C}{\mathcal{C}}
\newcommand{\D}{\mathcal{D}}
\newcommand{\dist}{\mathit{dist}}
\newcommand{\En}{\mathit{En}}
\newcommand{\F}{\mathcal{F}}
\newcommand{\kmax}{k_\mathit{max}}

\newcommand{\M}{\mathcal{M}}
\newcommand{\N}{\mathbb{N}}
\renewcommand{\P}{\mathcal{P}}
\renewcommand{\paragraph}[1]{\medskip \noindent \textbf{#1}}
\newcommand{\pmin}{p_\mathit{min}}
\newcommand{\Run}{\mathit{Run}}
\newcommand{\tran}[1]{\xrightarrow{#1}}
\newcommand{\val}{\mathit{val}}
\renewcommand{\vec}[1]{\mathbf{#1}}

\begin{document}

\title{The Odds of Staying on Budget}
\author{Christoph Haase\inst{1}\thanks{Supported by Labex Digicosme, Univ. Paris-Saclay, project VERICONISS.} \and Stefan Kiefer\inst{2}}
\institute{
Laboratoire Sp\'ecification et V\'erification (LSV), CNRS \& ENS de Cachan, France
\and
Department of Computer Science, University of Oxford, UK
}

\maketitle

\begin{abstract}
Given Markov chains and Markov decision processes (MDPs) whose
transitions are labelled with non-negative integer costs, we study the
computational complexity of deciding whether the probability of paths
whose accumulated cost satisfies a Boolean combination of inequalities
exceeds a given threshold. For acyclic Markov chains, we show that
this problem is PP-complete, whereas it is hard for the
\textsc{PosSLP} problem and in \textsc{PSpace} for general Markov
chains. Moreover, for acyclic and general MDPs, we prove
\textsc{PSpace}- and \textsc{EXP}-completeness, respectively.
Our results have direct implications on the complexity of
computing reward quantiles in succinctly represented stochastic
systems.
\end{abstract}

\section{Introduction} \label{sec-introduction}

Computing the shortest path from $s$ to~$t$ in a directed graph is a
ubiquitous problem in computer science, so shortest-path algorithms
such as Dijkstra's algorithm are a staple for every computer
scientist. These algorithms work in polynomial time even if the edges
are weighted, so
questions of the following kind are easy to answer:
\begin{quote}
  (I) \emph{Is it possible to travel from Copenhagen to Kyoto in less
    than 15 hours?}
\end{quote}
From a complexity-theoretic point of view, even computing the length
of the shortest path lies in~NC, the class of problems with
``efficiently parallelisable'' algorithms.%
\footnote{The NC algorithm performs ``repeated squaring'' of the weight matrix
 in the $(\mathord{\max}, \mathord{+})$-algebra.}

The shortest-path problem becomes more intricate as soon as
uncertainties are taken into account.  For example, additional
information such as ``\emph{there might be congestion in Singapore, so
  the Singapore route will, with probability~10\%, trigger a delay of
  1 hour}'' naturally leads to questions of the following kind:
\begin{quote}
  (II) \emph{Is there a travel plan avoiding trips longer than 15
    hours with probability $\ge 0.9$?}
\end{quote}
\emph{Markov decision processes (MDPs)} are the established model to
formalise problems such as~(II).  In each \emph{state} of an MDP some
\emph{actions} are enabled, each of which is associated with a
probability distribution over outgoing \emph{transitions}.  Each
transition, in turn, determines the successor state and is equipped
with a non-negative ``weight''. The weight could be interpreted as
time, distance, reward, or---as in this paper---as \emph{cost}. For
another example, imagine the plan of a research project whose workflow
can be modelled by a directed weighted graph. In each project state
the investigators can hire a programmer, travel to collaborators,
acquire new equipment, etc., but each action costs money, and the
result (i.e., the next project state) is probabilistic. The objective
is to meet the goals of the project before exceeding its budget for
the total accumulated cost.
This leads to questions such as:
\begin{quote}
  (III) \emph{Is there a strategy to stay on budget with probability
    $\ge 0.85$?}
\end{quote}
MDP problems like (II)~and~(III) become even more challenging when
each transition is equipped with both a cost and a \emph{utility},
e.g.\ in order to model problems that aim at maximising the probability
that both a given budget is kept \emph{and} a minimum total utility is
achieved. Such \emph{cost-utility trade-offs} have recently been
studied in~\cite{BDKquantiles2014}.

The problems (II)~and~(III) may become easier if there is no
non-determinism, i.e., there are no actions. We then obtain
\emph{Markov chains} where the next state and the incurred transition
cost are chosen in a purely probabilistic fashion. Referring to the
project example above, the activities may be completely planned out,
but their effects (i.e.\ cost and next state) may still be
probabilistic, yielding problems of the kind:
\begin{quote}
 (IV) \emph{Will the budget be kept with probability $\ge 0.85$?}
\end{quote}
Closely related to the aforementioned decision problems is the
following optimisation problem, referred to as the \emph{quantile
  query}
in~\cite{BDDKKquantiles2014,BDKquantiles2014,UmmelsBaierQuantiles}. A
quantile query asked by a funding body, for instance, could be the
following:
\begin{quote}
  (V) \emph{Given a probability threshold~$\tau$, compute the smallest
    budget that suffices with probability at least~$\tau$.}
\end{quote}

Non-stochastic problems like~(I) are well understood. The purpose of
this paper is to investigate the complexity of MDP problems such as
(II) and (III), of Markov-chain problems such as~(IV), and of quantile
queries like~(V). More formally, the models we consider are Markov
chains and MDPs with non-negative integer costs, and the main focus of
this paper is on the \emph{cost problem} for those models: Given a
budget constraint $\varphi$ represented as a Boolean combination of
linear inequalities and a probability threshold $\tau$, we study the
complexity of determining whether the probability of paths reaching a
designated target state with cost consistent with $\varphi$ is at
least $\tau$.

In order to highlight and separate our problems more clearly from
those in the literature, let us briefly discuss two
approaches that do not, at least not in an obvious way, resolve the
core challenges. First, one approach to answer the MDP problems could
be to compute a strategy that minimises the \emph{expected} total
cost, which is a classical problem in the MDP literature, solvable in
polynomial time using linear programming
methods~\cite{book:Puterman}. However, minimising the expectation may
not be optimal: if you don't want to be late, it may be better to walk
than to wait for the bus, even if the bus saves you time in average.
The second approach with shortcomings is to phrase problems (II),
(III) and~(IV) as MDP or Markov-chain \emph{reachability} problems,
which are also known to be solvable in polynomial time. This, however,
ignores the fact that numbers representing cost are commonly
represented in their natural succinct \emph{binary}
encoding. Augmenting each state with possible accumulated costs leads
to a blow-up of the state space which is exponential in the
representation of the input, giving an \textsc{EXP} upper bound as
in~\cite{BDKquantiles2014}.

\paragraph{Our contribution.}
The goal of this paper is to comprehensively investigate under which
circumstances and to what extent the complexity of the cost
problem and of quantile queries may be below the \textsc{EXP} upper bound.
We also provide new lower bounds, much stronger than the best known NP lower
bound derivable from~\cite{LSdurational05}.
We distinguish between acyclic and general control graphs.
In short, we show that the cost problem is
\begin{itemize}
\item PP-complete for acyclic Markov chains, and hard for the
  \textsc{PosSLP} problem and in \textsc{PSpace} in the general case;
  and
\item \textsc{PSpace}-complete for acyclic MDPs, and
  \textsc{EXP}-complete for general MDPs.
\end{itemize}
\iftechrep{}{Due to space constraints, we can only present the ideas
  underlying the proofs of our results; full details can be found in
  the technical report accompanying this paper~\cite{HK-ICALP15-TR}.}


\paragraph{Related Work.}
The motivation for this paper comes from the work on quantile queries
in~\cite{BDDKKquantiles2014,BDKquantiles2014,UmmelsBaierQuantiles}
mentioned above and on model checking so-called durational
probabilistic systems~\cite{LSdurational05} with a probabilistic timed
extension of CTL. 
While the focus of~\cite{UmmelsBaierQuantiles} is mainly on
``qualitative'' problems where the probability threshold is either 0
or~1, an iterative linear-programming-based approach for solving
quantile queries has been suggested
in~\cite{BDDKKquantiles2014,BDKquantiles2014}. The authors report
satisfying experimental results, the worst-case complexity however
remains exponential time. Settling the complexity of quantile queries
has been identified as one of the current challenges in the conclusion
of~\cite{BDKquantiles2014}.

Recently, there has been considerable interest in models of stochastic
systems that extend weighted graphs or counter systems,
see~\cite{VarSSP15} for a very recent survey.
Multi-dimensional percentile queries for various payoff functions are studied in~\cite{RRS-CAV15}.
The work by
Bruy{\`e}re et al.~\cite{BFRRStacs14} has also been motivated by the
fact that minimising the expected total cost is not always an adequate
solution to natural problems. For instance, they consider the problem
of computing a scheduler in an MDP with positive integer weights that
ensures that both the expected and the maximum incurred cost remain
below a given values. Other recent work also investigated MDPs with a
single counter ranging over the non-negative integers, see
e.g.~\cite{BBEK13,EBKBWSoda10}. However, in that work updates to the
counter can be both positive and negative. For that reason, the
analysis focuses on questions about the counter value \emph{zero},
such as designing a strategy that maximises the probability of
reaching counter value zero.



\section{Preliminaries} \label{sec-prelim}

We write $\N = \{0, 1, 2, \ldots\}$.
For a countable set $X$ we write $\dist(X)$ for the set of \emph{probability distributions} over~$X$;
 i.e., $\dist(X)$ consists of those functions $f : X \to [0,1]$ such that $\sum_{x \in X} f(x) = 1$.

\paragraph{Markov Chains.} A \emph{Markov chain} is a triple $\M = (S,s_0,\delta)$,
where $S$ is a countable (finite or infinite) set of states, $s_0 \in S$ is an initial state,
and $\delta: S \to \dist(S)$ is a probabilistic transition function
 that maps a state to a probability distribution over the successor states.
Given a Markov chain we also write $s \tran{p} t$ or $s \tran{} t$ to indicate that $p = \delta(s)(t) > 0$.
A \emph{run} is an infinite sequence $s_0 s_1 \cdots \in \{s_0\} S^\omega$
 with $s_i \tran{} s_{i+1}$ for $i \in \N$.
We write $\Run(s_0 \cdots s_k)$ for the set of runs that start with $s_0 \cdots s_k$.
To~$\M$ we associate the standard probability space $(\Run(s_0),\F,\P)$
where $\F$ is the $\sigma$-field generated by all basic cylinders
 $\Run(s_0 \cdots s_k)$ with $s_0 \cdots s_k \in \{s_0\} S^*$,
and $\P: \F \to [0,1]$ is the unique probability measure such that $\P(\Run(s_0 \cdots s_k)) = \prod_{i=1}^{k} \delta(s_{i-1})(s_i)$.

\paragraph{Markov Decision Processes.}
A \emph{Markov decision process (MDP)} is a tuple $\D = (S, s_0, A, \En, \delta)$,
 where $S$ is a countable set of states, $s_0 \in S$ is the initial state,
 $A$ is a finite set of actions,
 $\En : S \to 2^A \setminus \emptyset$ is an action enabledness function that assigns to each state~$s$ the set $\En(s)$ of actions enabled in~$s$,
 and $\delta : S \times A \to \dist(S)$ is a probabilistic transition function that maps a state~$s$ and an action $a \in \En(s)$ enabled in~$s$
 to a probability distribution over the successor states.
A (deterministic, memoryless) \emph{scheduler} for~$\D$ is a function $\sigma : S \to A$
 with $\sigma(s) \in \En(s)$ for all $s \in S$.
A scheduler~$\sigma$ induces a Markov chain $\M_\sigma = (S, s_0, \delta_\sigma)$
 with $\delta_\sigma(s) = \delta(s,\sigma(s))$ for all $s \in S$.
We write $\P_\sigma$ for the corresponding probability measure of~$\M_\sigma$.

\paragraph{Cost Processes.}
A \emph{cost process} is a tuple $\C = (Q, q_0, t, A, \En, \Delta)$,
where $Q$ is a finite set of control states, $q_0 \in Q$ is the
initial control state, $t$ is the target control state, $A$ is a
finite set of actions, $\En : Q \to 2^A \setminus \emptyset$ is an
action enabledness function that assigns to each control state $q$ the
set $\En(q)$ of actions enabled in~$q$, and $\Delta : Q \times A \to
\dist(Q \times \N)$ is a probabilistic transition function.  Here, for
$q, q' \in Q$, $a \in \En(q)$ and $k \in \N$, the value $\Delta(q,
a)(q',k) \in [0,1]$ is the probability that, if action~$a$ is taken in
control state~$q$, the cost process transitions to control state~$q'$
and cost~$k$ is incurred.  For the complexity results we define the
\emph{size} of~$\C$ as the size of a succinct description, i.e., the
costs are encoded in binary, the probabilities are encoded as
fractions of integers in binary (so the probabilities are rational),
and for each $q \in Q$ and $a \in
\En(q)$, the distribution $\Delta(q,a)$ is described by the list of
triples $(q',k,p)$ with $\Delta(q,a)(q',k) = p > 0$ (so we
assume this list to be finite).  Consider the
directed graph $G=(Q,E)$ with
 \[
  E := \{ (q,q') \in (Q \setminus \{t\}) \times Q : \exists a \in \En(q) \ \exists k \in \N.\ \Delta(q,a)(q',k) > 0\} \;.
 \]
We call $\C$ \emph{acyclic} if $G$ is acyclic (which can be determined in linear time).

A cost process~$\C$ induces an MDP $\D_\C = (Q \times \N, (q_0,0), A,
\En', \delta)$ with $\En'(q,c) = \En(q)$ for all $q \in Q$ and $c \in
\N$, and $ \delta((q,c),a)(q',c') = \Delta(q,a)(q',c'-c) $ for all
$q,q' \in Q$ and $c, c' \in \N$ and $a \in A$.  For a state $(q,c) \in
Q \times \N$ in~$\D_\C$ we view~$q$ as the current control state and
$c$ as the current cost, i.e., the cost accumulated thus far. We refer
to~$\C$ as a \emph{cost chain} if $|\En(q)| = 1$ holds for all $q \in
Q$. In this case one can view $\D_\C$ as the Markov chain induced by
the unique scheduler of~$\D_\C$.  For cost chains, actions are not
relevant, so we describe cost chains just by the tuple $\C = (Q, q_0,
t, \Delta)$.

Recall that we restrict schedulers to be deterministic and memoryless,
as such schedulers will be sufficient for the objectives in this
paper.  Note, however, that our definition allows schedulers to depend
on the current cost, i.e., we may have schedulers~$\sigma$ with
$\sigma(q,c) \ne \sigma(q,c')$.

\paragraph{The accumulated cost~$K$.}
In this paper we will be interested in the cost accumulated during a run before reaching the target state~$t$.
For this cost to be a well-defined random variable,
we make two assumptions on the system:
(i)~We assume that $\En(t) = \{a\}$ holds for some $a \in A$ and
$\Delta(t,a)(t,0) = 1$.  Hence, runs that visit~$t$ will not leave~$t$
and accumulate only a finite cost.  (ii)~We assume that for all
schedulers the target state~$t$ is almost surely reached, i.e., for
all schedulers the probability of eventually visiting a state $(t,c)$
with $c \in \N$ is equal to one.  The latter condition can be verified
by graph algorithms in time quadratic in the input size, e.g., by
computing the \emph{maximal end components} of the MDP obtained
from~$\C$ by ignoring the cost, see
e.g.\ \cite[Alg.~47]{BaierKatoen-book}.

Given a cost process~$\C$ 
 we define a random variable $K_\C : \Run((q_0,0)) \to \N$ such that $K_\C((q_0,0) \ (q_1,c_1) \ \cdots) = c$
 if there exists $i \in \N$ with $(q_i,c_i) = (t,c)$.
We often drop the subscript from~$K_\C$ if the cost process~$\C$ is clear from the context.
We view $K(w)$ as the accumulated cost of a run~$w$.

From the above-mentioned assumptions on~$t$, it follows that for any
scheduler the random variable~$K$ is almost surely defined.  Dropping
assumption~(i) would allow the same run to visit states $(t,c_1)$ and
$(t,c_2)$ for two different $c_1, c_2 \in \N$.  There would still be
reasonable ways to define a cost~$K$, but no apparently best way. If
assumption~(ii) were dropped, we would have to deal with runs that do
not visit the target state~$t$.  In that case one could study the
random variable~$K$ as above \emph{conditioned} under the event that
$t$ is visited. For Markov chains, \cite[Sec.~3]{BKKM14conditional}
describes a transformation that preserves the distribution of the
conditional cost~$K$, but $t$ is almost surely reached in the
transformed Markov chain.  In this sense, our assumption~(ii) is
without loss of generality for cost chains.  For general cost
processes the transformations of~\cite{BKKM14conditional} do not work.
In fact, a scheduler that ``optimises''~$K$ conditioned under
reaching~$t$ might try to avoid reaching~$t$ once the accumulated cost
has grown unfavourably.  Hence, dropping assumption~(ii) in favour of
conditional costs would give our problems an aspect of multi-objective
optimisation, which is not the focus of this paper.

\paragraph{The cost problem.}
Let $x$ be a fixed variable.
An \emph{atomic cost formula} is an inequality of the form $x \le B$ where $B \in \N$ is encoded in binary.
A \emph{cost formula} is an arbitrary Boolean combination of atomic cost formulas.
A number $n \in \N$ \emph{satisfies} a cost formula~$\varphi$, in symbols $n \models \varphi$,
 if $\varphi$ is true when $x$ is replaced by~$n$.

This paper mainly deals with the following decision problem: given a
cost process~$\C$, a cost formula~$\varphi$, and a probability
threshold $\tau \in [0,1]$, the \emph{cost problem} asks whether there
exists a scheduler~$\sigma$ with $\P_\sigma(K_\C \models \varphi) \ge
\tau$.  The case of an atomic cost formula~$\varphi$ is an important
special case.  Clearly, for cost chains~$\C$ the cost problem simply
asks whether $\P(K_\C \models \varphi) \ge \tau$ holds.  One can
assume $\tau = 1/2$ without loss of generality, thanks to a simple
construction, see \iftechrep{Prop.~\ref{prop-a-half} in App.~\ref{app-prelim}}{\cite{HK-ICALP15-TR}}.
Moreover, with an oracle for the cost problem at hand, one can use
binary search over~$\tau$ to approximate $\P_\sigma(K \models
\varphi)$: $i$ oracle queries suffice to approximate $\P_\sigma(K
\models \varphi)$ within an absolute error of $2^{-i}$.

By our definition, the MDP~$\D_\C$ is in general infinite as there is
no upper bound on the accumulated cost.  However, when solving the
cost problem, there is no need to keep track of costs above~$B$, where
$B$ is the largest number appearing in~$\varphi$.  So one can solve
the cost problem in so-called \emph{pseudo-polynomial time} (i.e.,
polynomial in~$B$, not in the size of the encoding of~$B$) by
computing an explicit representation of a restriction, say
$\widehat{\D}_\C$, of~$\D_\C$ to costs up to~$B$, and then applying
classical linear-programming techniques \cite{book:Puterman} to
compute the optimal scheduler for the finite MDP~$\widehat{\D}_\C$.
Since we consider reachability objectives, the optimal scheduler is
deterministic and memoryless.  This shows that our restriction to
deterministic memoryless schedulers is without loss of generality.  In
terms of our succinct representation we have:
\begin{proposition} \label{fact-EXPTIME-upper}
 The cost problem is in \textsc{EXP}.
\end{proposition}
Heuristic improvements to this approach were suggested in~\cite{UmmelsBaierQuantiles,BDDKKquantiles2014}.
The subject of this paper is to investigate to what extent the \textsc{EXP} complexity is optimal.

\section{Quantile Queries} \label{sec-quantile}

In this section we consider the following function problem, referred to as
\emph{quantile query} in \cite{UmmelsBaierQuantiles,BDDKKquantiles2014,BDKquantiles2014}.
Given a cost chain~$\C$ and a probability threshold~$\tau$,
 a quantile query asks for the smallest budget~$B$ such that $\P_\sigma(K_\C \le B) \ge \tau$.
We show that polynomially many oracle queries to the cost problem for atomic cost formulas ``$x \le B$'' suffice to answer a quantile query.
This can be done using binary search over the budget~$B$.
The following proposition\iftechrep{, proved in App.~\ref{app-quantile},}{}
 provides a suitable general upper bound on this binary search,
 by exhibiting a concrete sufficient budget, computable in polynomial time:
\newcommand{\stmtpropaprioriupperbound}{
 Suppose $0 \le \tau < 1$.
 Let $\pmin$ be the smallest non-zero probability and $\kmax$ be the largest cost
  in the description of the cost process. Then $\P_\sigma(K \le B) \ge \tau$ holds for all schedulers~$\sigma$, where
  \[
   B := \kmax \cdot \left\lceil |Q| \cdot \left( -\ln (1-\tau) / \pmin^{|Q|} \ + \ 1 \right) \right\rceil\;.
  \]
}
\begin{proposition} \label{prop-apriori-upper-bound}
\stmtpropaprioriupperbound
\end{proposition}
The case $\tau = 1$ is covered by~\cite[Thm.~6]{UmmelsBaierQuantiles},
 where it is shown that one can compute in polynomial time the smallest~$B$ with
  $\P_\sigma(K \le B) = 1$ for all schedulers~$\sigma$, if such $B$ exists.
We conclude that quantile queries are polynomial-time inter-reducible with the cost problem for atomic cost formulas.

\section{Cost Chains} \label{sec-MC}
In this section we consider the cost problems for acyclic and general
cost chains. Even in the general case we obtain \textsc{PSpace}
membership, avoiding the \textsc{EXP} upper bound from
Prop.~\ref{fact-EXPTIME-upper}.

\paragraph{Acyclic Cost Chains.} \label{sub-MC-acyclic}
The complexity class PP~\cite{Gill77:PP} can be defined as the class
of languages~$L$ that have a probabilistic polynomial-time bounded
Turing machine~$M_L$ such that for all words~$x$ one has $x \in L$ if
and only if $M_L$ accepts~$x$ with probability at least~$1/2$. The
class~PP includes~NP~\cite{Gill77:PP},
and Toda's theorem states that P$^{\text{PP}}$ contains the polynomial-time
hierarchy~\cite{Toda91}.
We show that the cost problem for acyclic cost chains is PP-complete.

\newcommand{\stmtthmMCacyclic}{The cost problem for acyclic cost
  chains is in PP. It is PP-hard under polynomial-time Turing
  reductions, even for atomic cost formulas.}
\begin{theorem} \label{thm-MC-acyclic}
\stmtthmMCacyclic
\end{theorem}
\begin{proof}[sketch]
  To show membership in~PP,
  we construct a probabilistic Turing machine that simulates the acyclic cost chain, and keeps track of the currently accumulated cost on the tape.
  For the lower bound, it follows from~\cite[Prop.~4]{LSdurational05}
  that an instance of the \textsc{$K$th largest subset} problem can be
  reduced to a cost problem for acyclic cost chains with atomic cost
  formulas.
  We show in~\cite[Thm.~3]{HK15} that this problem is PP-hard under  polynomial-time Turing reductions.
\qed
\end{proof}
PP-hardness strengthens the NP-hardness result
from~\cite{LSdurational05} substantially:
by Toda's theorem it follows
that any problem in the polynomial-time hierarchy can be solved by a
deterministic polynomial-time bounded Turing machine that has oracle
access to the cost problem for acyclic cost chains.

\paragraph{General Cost Chains.} \label{sub-MC-cyclic}
For the PP upper bound in Thm.~\ref{thm-MC-acyclic},
 the absence of cycles in the control graph seems essential.
Indeed, we can use cycles to show hardness for the \textsc{PosSLP} problem,
suggesting that the acyclic and the general case have different complexity.
\textsc{PosSLP} is a fundamental problem for numerical computation~\cite{Allender09numericalAnalysis}.
Given 
an arithmetic circuit
with operators $\mathord{+}$, $\mathord{-}$,
$\mathord{*}$, inputs 0 and 1, and a designated output gate,
the \textsc{PosSLP} problem asks whether
the circuit outputs a positive integer.
\textsc{PosSLP} is in \textsc{PSpace}; in fact, it lies in
the 4th level of the \emph{counting hierarchy
  (CH)}~\cite{Allender09numericalAnalysis}, an analogue to the
polynomial-time hierarchy for classes like~PP.
We have the following theorem:

\newcommand{\stmtthmMCcyclic}{ The cost problem for cost chains is in
  \textsc{PSpace} and hard for \textsc{PosSLP}.  }
\begin{theorem}\label{thm-MC-cyclic}
\stmtthmMCcyclic
\end{theorem}
The remainder of this section is devoted to a proof sketch of this
theorem. Showing membership in \textsc{PSpace} requires non-trivial
results. There is no agreed-upon definition of probabilistic
\textsc{PSpace} in the literature, but we can define it in analogy
to~PP as follows: \emph{Probabilistic \textsc{PSpace}} is the class of
languages~$L$ that have a probabilistic polynomial-space bounded
Turing machine~$M_L$ such that for all words~$x$ one has $x \in L$ if
and only if $M_L$ accepts~$x$ with probability at least~$1/2$. The
cost problem for cost chains is in this class, as can be shown by
adapting the argument from the beginning of the proof sketch for
Thm.~\ref{thm-MC-acyclic}, replacing PP with probabilistic
\textsc{PSpace}. It was first proved in~\cite{Simon77} that
probabilistic \textsc{PSpace} equals \textsc{PSpace}, hence the cost
problem for cost chains is in \textsc{PSpace}.

For the \textsc{PosSLP}-hardness proof
one can assume the following normal form,
 see the proof of~\cite[Thm. 5.2]{EY09rmc}:
there are only $\mathord{+}$ and~$\mathord{\ast}$ operators,
the corresponding gates alternate,
and all gates except those on the bottom level have exactly two incoming edges,
cf.\ the top of~Fig.~\ref{fig-reduction}.
We write $\mathit{val}(g)$ for the value output by gate~$g$.
Then \textsc{PosSLP} asks:
given an arithmetic circuit (in normal form) including gates $g_1, g_2$,
 is $\mathit{val}(g_1) \ge \mathit{val}(g_2)$?

\begin{figure}[t]
  \begin{center}
    \includegraphics[scale=1]{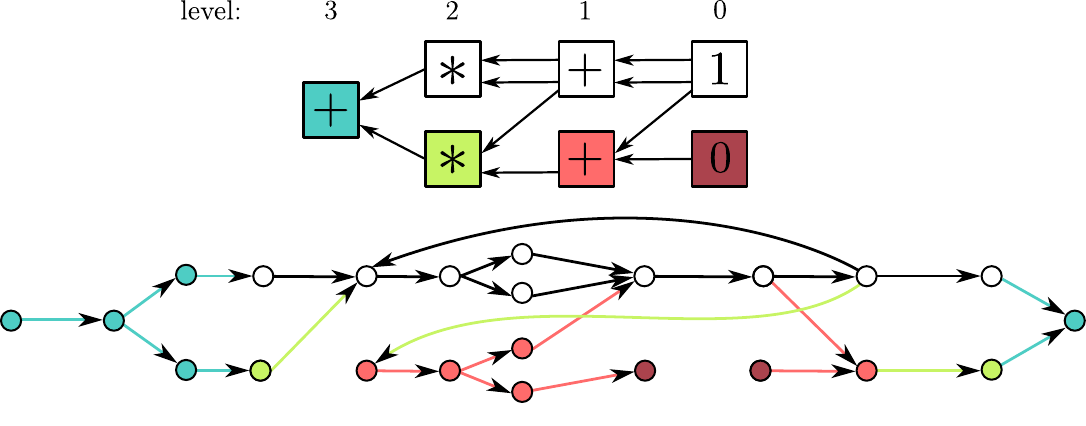}
  \end{center}
  \caption{
Top: an arithmetic circuit in normal form.
Bottom: a DFA (omitting input letters) corresponding to the construction of Prop.~\ref{lem:parikh-hardness}.
Identical colours indicate a correspondence between gates and states.
}
  \label{fig-reduction}
\end{figure}

As an intermediate step of independent interest, we show
\textsc{PosSLP}-hardness of a problem about deterministic finite
automata (DFAs).  Let $\Sigma$ be a finite alphabet and call a
function $f: \Sigma \to \mathbb{N}$ a \emph{Parikh function}. The
\emph{Parikh image} of a word $w\in \Sigma^*$ is the Parikh
function~$f$ such that $f(a)$ is the number of occurrences of~$a$
in~$w$.  We show:
\begin{proposition}\label{lem:parikh-hardness}
Given an arithmetic circuit including gate~$g$,
 one can compute in logarithmic space
 a Parikh function~$f$ (in binary encoding) and a DFA~$\mathcal{A}$
 such that $\mathit{val}(g)$ equals
 the number of accepting computations in~$\mathcal{A}$ that are
 labelled with words that have Parikh image~$f$.
\end{proposition}
%
The construction is illustrated in Fig.~\ref{fig-reduction}.
It is by induction on the levels of the arithmetic circuit.
A gate labelled with ``+'' is simulated by \emph{branching}
into the inductively constructed gadgets corresponding to the gates
this gate connects to. Likewise, a gate labelled with ``$\ast$'' is
simulated by \emph{sequentially composing} the gadgets corresponding
to the gates this gate connects to.
It is the latter case that may introduce cycles in the structure of the DFA.
Building on this construction,
by encoding alphabet letters in natural numbers encoded in binary,
we then show:
\begin{proposition}\label{prop-SLPtoMC}
Given an arithmetic circuit including gate~$g$ on odd level~$\ell$,
 one can compute in logarithmic space 
 a cost process~$\C$ and $T \in \N$ with $\P(K_\C = T) = \mathit{val}(g)/m$, 
 where $m = \exp_2(2^{(\ell-1)/2+1}-1)\cdot \exp_d(2^{(\ell-1)/2+1}-3)$.
\end{proposition}
Towards the \textsc{PosSLP} lower bound from Thm.~\ref{thm-MC-cyclic},
given an arithmetic circuit including gates $g_1, g_2$,
 we use Prop.~\ref{prop-SLPtoMC} to construct two cost chains
 $\C_1 = (Q, q_1, t, \Delta)$ and $\C_2 = (Q, q_2, t, \Delta)$ and
 $T_1,T_2 \in \N$ such that $\P(K_{\C_i} = T_i) = \mathit{val}(g_i)/m$
 holds for $i \in \{1,2\}$ and for $m \in \N$ as in Prop.~\ref{prop-SLPtoMC}.
Then we compute a number $H \ge T_2$ such that $\P(K_{\C_2} > H) < 1/m$.
The representation of~$m$ from Prop.~\ref{prop-SLPtoMC} is of exponential size.
However, using Prop.~\ref{prop-apriori-upper-bound},
 $H$ depends only logarithmically on~$m+1$.
We combine $\C_1$ and~$\C_2$ to a cost chain $\C = (Q \uplus \{q_0\},
q_0, t, \widetilde\Delta)$, where $\widetilde\Delta$ extends~$\Delta$
by $\widetilde\Delta(q_0)(q_1,H+1) = 1/2$ and
$\widetilde\Delta(q_0)(q_2,0) = 1/2$.  By this construction, the new
cost chain~$\C$ initially either incurs cost~$H+1$ and then
emulates~$\C_1$, or incurs cost~$0$ and then emulates~$\C_2$.  Those
possibilities have probability $1/2$ each.

Finally, we compute a suitable cost formula~$\varphi$ such that we
have $\mathit{val}(g_1) \ge \mathit{val}(g_2)$ if and only if $\P(K_\C
\models \varphi) \ge 1/2$, completing the logspace reduction.  We
remark that the structure of the formula~$\varphi$, in particular the
number of inequalities, is fixed.  Only the involved numbers depend on
the concrete instance.


\section{Cost Processes} \label{sec-MDP}
\paragraph{Acyclic Cost Processes.} \label{sub-MDP-acyclic}
We now prove that the cost problem for acyclic cost processes is
\textsc{PSpace}-complete. The challenging part is to show that
\textsc{PSpace}-hardness even holds for \emph{atomic} cost
formulas. For our lower bound, we reduce from a generalisation of the
classical \textsc{SubsetSum} problem: Given a tuple $(k_1, \ldots,
k_n, T)$ of natural numbers with $n$ even, the
\emph{\textsc{QSubsetSum} problem} asks whether the following formula
is true:
 \[
  \exists x_1 \in \{0,1\} \ \forall x_2 \in \{0,1\} \ \cdots \ \exists x_{n-1} \in \{0,1\} \ \forall x_n \in \{0,1\} \ : \ \sum_{1\le i\le n} x_i k_i = T
 \]
Here, the quantifiers $\mathord{\exists}$ and~$\mathord{\forall}$
occur in strict alternation. It is shown in~\cite[Lem.~4]{Travers06}
that \textsc{QSubsetSum} is \textsc{PSpace}-complete. One can think of
such a formula as a turn-based game, the \emph{\textsc{QSubsetSum}
  game}, played between Player Odd and Player Even.  If $i \in \{1,
\ldots, n\}$ is odd (even), then turn~$i$ is Player Odd's (Player
Even's) turn, respectively.  In turn~$i$ the respective player decides
to either \emph{take}~$k_i$ by setting $x_i=1$, or \emph{not} to
\emph{take}~$k_i$ by setting $x_i=0$.  Player Odd's objective is to
make the sum of the taken numbers equal~$T$, and Player Even tries to
prevent that.  If Player Even is replaced by a random player, then
Player Odd has a strategy to win with probability~$1$ if and only if
the given instance is a ``yes'' instance for \textsc{QSubsetSum}.
This gives an easy \textsc{PSpace}-hardness proof for the cost problem
with non-atomic cost formulas $\varphi \equiv (x=T)$.  In order to
strengthen the lower bound to atomic cost formulas $\varphi \equiv (x
\le B)$ we have to give Player Odd an incentive to take numbers~$k_i$,
although she is only interested in not exceeding the budget~$B$. This
challenge is addressed in our \textsc{PSpace}-hardness proof.

The \textsc{PSpace}-hardness result reflects the fact that the optimal
strategy must take the current cost into account,
 not only the control state, even for atomic cost formulas.
This may be somewhat counter-intuitive, as a good strategy
should always ``prefer small cost''.  But if there always existed a
strategy depending \emph{only} on the control state, one could guess
this strategy in~NP and invoke the PP-result of
Sec.~\ref{sub-MC-acyclic} in order to obtain an NP$^\text{PP}$
algorithm, implying NP$^\text{PP}$ = \textsc{PSpace} and hence a
collapse of the counting hierarchy.

Indeed, for a concrete example, consider the acyclic cost process with $Q =
\{q_0, q_1, t\}$, and $\En(q_0) = \{a\}$ and $\En(q_1) = \{a_1,
a_2\}$, and $\Delta(q_0,a)(q_1,+1) = \frac12$ and
$\Delta(q_0,a)(q_1,+3) = \frac12$ and $\Delta(q_1,a_1)(t,+3) = 1$ and
$\Delta(q_1,a_2)(t,+6) = \frac12$ and $\Delta(q_1,a_2)(t,+1) =
\frac12$.
Consider the atomic cost formula $\varphi \equiv (x \le 5)$.
An optimal scheduler~$\sigma$ plays~$a_1$ in~$(q_1,1)$ and $a_2$
in~$(q_1,3)$, because additional cost~$3$, incurred by~$a_1$, is fine
in the former but not in the latter configuration.  For this
scheduler~$\sigma$ we have $\P_\sigma(K \models \varphi) = \frac34$.

\newcommand{\stmtthmMDPacyclic}{
 The cost problem for acyclic cost processes is in \textsc{PSpace}.
 It is \textsc{PSpace}-hard, even for atomic cost formulas.
}
\begin{theorem} \label{thm-MDP-acyclic}
\stmtthmMDPacyclic
\end{theorem}
\newcommand{\Opt}{\textsc{Opt}}
\begin{proof}[sketch]
To prove membership in \textsc{PSpace}, we consider a procedure
\Opt\ that, given $(q,c) \in Q \times \N$ as input, computes the
optimal (i.e., maximised over all schedulers) probability~$p_{q,c}$
that starting from~$(q,c)$ one reaches~$(t,d)$ with $d \models
\varphi$.  The following procedure characterisation of~$p_{q,c}$ for
$q \ne t$ is crucial for \Opt$(q,c)$:
 \[
  p_{q,c} = \max_{a \in \En(q)} \sum_{q' \in Q}
  \sum_{k \in \N} \Delta(q,a)(q',k) \cdot p_{q',c+k}
 \]
 So \Opt$(q,c)$ loops over all $a \in \En(q)$ and all $(q',k) \in Q
 \times \N$ with $\Delta(q,a)(q',k) > 0$ and recursively computes
 $p_{q',c+k}$.  Since the cost process is acyclic, the height of the
 recursion stack is at most~$|Q|$.  The representation size of the
 probabilities that occur in that computation is polynomial.  To see
 this, consider the product~$D$ of the denominators of the
 probabilities occurring in the description of~$\Delta$.  The encoding
 size of~$D$ is polynomial.  All probabilities occurring during the
 computation are integer multiples of~$1/D$.  Hence
 computing~$\Opt(q_0,0)$ and comparing the result with~$\tau$ gives a
 \textsc{PSpace} procedure.

 For the lower bound we reduce the \textsc{QSubsetSum} problem,
 defined above, to the cost problem for an atomic cost formula $x \le
 B$.  Given an instance $(k_1, \ldots, k_n, T)$ with $n$ is even of
 the \textsc{QSubsetSum} problem, we construct an acyclic cost process
 $\C = (Q, q_0, t, A, \En, \Delta)$ as follows.  We take $Q = \{q_0,
 q_2, \ldots, q_{n-2}, q_n, t\}$.  Those control states reflect pairs
 of subsequent turns that the \textsc{QSubsetSum} game can be in.  The
 transition rules~$\Delta$ will be set up so that probably the control
 states $q_0, q_2, \ldots, q_n, t$ will be visited in that order, with
 the (improbable) possibility of shortcuts to~$t$.  For even~$i$ with
 $0 \le i \le n{-}2$ we set $\En(q_i) = \{a_0, a_1\}$.  These actions
 correspond to Player Odd's possible decisions of not taking,
 respectively taking~$k_{i+1}$.  Player Even's response is modelled by
 the random choice of not taking, respectively taking~$k_{i+2}$ (with
 probability $1/2$ each).  In the cost process, taking a number~$k_i$
 corresponds to incurring cost~$k_i$.  We also add an additional
 cost~$\ell$ in each transition.%
 \footnote{ This is for technical reasons. Roughly speaking, this
   prevents the possibility of reaching the full budget~$B$ before an
   action in control state~$q_{n{-}2}$ is played.}  Therefore we
 define our cost problem to have the atomic formula $x \le B$ with $B
 := (n/2) \cdot \ell + T$. For some large number $M \in \N$, formally
 defined \iftechrep{in App.~\ref{app-MDP}}{in~\cite{HK-ICALP15-TR}}, we set for all even $i \le n{-}2$ and
 for $j \in \{0,1\}$:
 \begin{align*}
  \Delta(q_i,a_j)(q_{i+2},\ell+j \cdot k_{i+1})           &= (1/2) \cdot \left(1 - (\ell+j \cdot k_{i+1})/M \right) \\
  \Delta(q_i,a_j)(t,\ell+j \cdot k_{i+1})                 &= (1/2) \cdot (\ell+j \cdot k_{i+1})/M \\
  \Delta(q_i,a_j)(q_{i+2},\ell+j \cdot k_{i+1}+k_{i+2}) &= (1/2) \cdot \left(1 - (\ell+j \cdot k_{i+1}+k_{i+2})/M \right) \\
  \Delta(q_i,a_j)(t,\ell+j \cdot k_{i+1}+k_{i+2})       &= (1/2) \cdot (\ell+j \cdot k_{i+1}+k_{i+2})/M%
 \intertext{%
   So with high probability the MDP transitions from~$q_i$ to~$q_{i+2}$,
  and cost $\ell$, $\ell+k_{i+1}$, $\ell+k_{i+2}$, $\ell+k_{i+1}+k_{i+2}$ is incurred,
  depending on the scheduler's (i.e., Player Odd's) actions and on the random (Player Even) outcome.
 But with a small probability, which is proportional to the incurred cost,
  the MDP transitions to~$t$, which is a ``win'' for the scheduler
   as long as the accumulated cost is within budget~$B$.
 We make sure that the scheduler loses if $q_n$ is reached:
 }
  \Delta(q_n,a)(t,B{+}1) &= 1 \qquad \text{with $\En(q_n) = \{a\}$}
 \end{align*}
The MDP is designed so that the scheduler probably ``loses'' (i.e., exceeds the budget~$B$);
 but whenever cost~$k$ is incurred, a winning opportunity with probability~$k/M$ arises.
Since $1/M$ is small, the overall probability of winning is approximately $C/M$ if total cost $C \le B$ is incurred.
In order to maximise this chance, the scheduler wants to maximise the total cost without exceeding~$B$,
 so the optimal scheduler will target~$B$ as total cost.

The values for $\ell$, $M$ and~$\tau$ need to be chosen carefully, as
the overall probability of winning is not exactly the sum of the
probabilities of the individual winning opportunities.  By the ``union
bound'', this sum is only an upper bound, and one needs to show that
the sum approximates the real probability closely enough. \qed
\end{proof}

\paragraph{General Cost Processes.}
\label{sub-MDP-general}
We show the following theorem:
\newcommand{\stmtthmMDPgeneral}{
 The cost problem is \textsc{EXP}-complete.
}
\begin{theorem} \label{thm-MDP-general}
\stmtthmMDPgeneral
\end{theorem}
\begin{proof}[sketch]
The \textsc{EXP} upper bound was stated in
Prop.~\ref{fact-EXPTIME-upper}.  Regarding hardness, we build upon
\emph{countdown games}, ``a simple class of turn-based 2-player games with discrete timing'' \cite{JurdzinskiSL08}.
Deciding the winner in a countdown game is
\textsc{EXP}-complete~\cite{JurdzinskiSL08}. Albeit non-stochastic,
countdown games are very close to our model: two players move along
edges of a graph labelled with positive integer weights and thereby
add corresponding values to a succinctly encoded counter. Player~1's
objective is to steer the value of the counter to a given number $T
\in \N$, and Player~2 tries to prevent that.  Our reduction from
countdown games \iftechrep{in App.~\ref{app-MDP}}{in~\cite{HK-ICALP15-TR}} requires a small trick, as in
our model the final control state~$t$ needs to be reached with
probability~1 regardless of the scheduler, and furthermore, the
scheduler attempts to achieve the cost target~$T$ when and only when
the control state $t \in Q$ is visited.
\qed
\end{proof}

\paragraph{The Cost-Utility Problem.}
MDPs with \emph{two} non-negative and non-decreasing integer counters,
viewed as cost and utility, respectively, were considered in
\cite{BDDKKquantiles2014,BDKquantiles2014}. Specifically, those works
consider problems such as computing the minimal cost~$C$ such that the
probability of gaining at least a given utility~$U$ is at
least~$\tau$.  Possibly the most fundamental of those problems is the
following: the \emph{cost-utility problem} asks, given an MDP with
both cost and utility, and numbers $C, U \in \N$, whether one can,
with probability~1, gain utility at least~$U$ using cost at most~$C$.
Using essentially the proof of Thm.~\ref{thm-MDP-general} we show:
\newcommand{\stmtcorcostutility}{
The cost-utility problem is \textsc{EXP}-complete.
}
\begin{corollary} \label{cor-cost-utility}
\stmtcorcostutility
\end{corollary}


\label{sec-extension}
\paragraph{The Universal Cost Problem.} We defined the cost problem so that it asks whether \emph{there exists} a scheduler~$\sigma$ with
$\P_\sigma(K_\C \models \varphi) \ge
\tau$.
A natural variant is the \emph{universal cost problem}, which asks
whether \emph{for all} schedulers~$\sigma$ we have
$\P_\sigma(K_\C \models \varphi) \ge \tau$.  Here the scheduler is
viewed as an adversary which tries to prevent the satisfaction
of~$\varphi$.  Clearly, for cost chains the cost problem and the
universal cost problem are equivalent.  Moreover,
Thms.~\ref{thm-MDP-acyclic}~and~\ref{thm-MDP-general} hold
analogously in the universal case.
\newcommand{\stmtthmuniversal}{
The universal cost problem for acyclic cost processes is in \textsc{PSpace}.
It is \textsc{PSpace}-hard, even for atomic cost formulas.
The universal cost problem is \textsc{EXP}-complete.
}
\begin{theorem} \label{thm-universal}
\stmtthmuniversal
\end{theorem}
\begin{proof}[sketch]
The universal cost problem and the complement of the cost problem
(and their acyclic versions) are interreducible in logarithmic
space by essentially negating the cost-formulas.
The only problem is that if $\varphi$ is an \emph{atomic} cost
formula, then $\neg\varphi$ is not an atomic cost formula. However,
the \textsc{PSpace}-hardness proof from Thm.~\ref{thm-MDP-acyclic} can
be adapted, cf.\ \iftechrep{App.~\ref{app-extension}}{\cite{HK-ICALP15-TR}}.  \qed
\end{proof}


\section{Conclusions and Open Problems} \label{sec-conclusions}

In this paper we have studied the complexity of analysing succinctly represented stochastic systems
with a single non-negative and only increasing integer counter.
We have improved the known complexity bounds significantly.
Among other results,
we have shown that the cost problem for Markov chains is in
\textsc{PSpace} and both hard for PP and the \textsc{PosSLP} problem.
It would be fascinating and potentially challenging to prove either
\textsc{PSpace}-hardness or membership in the counting hierarchy:
the problem does not seem to lend itself to a
\textsc{PSpace}-hardness proof, but the authors are not aware of natural
problems, except \textsc{BitSLP}~\cite{Allender09numericalAnalysis},
that are in the counting hierarchy and known to be hard for both PP
and \textsc{PosSLP}.

Regarding acyclic and general MDPs, we have proved
\textsc{PSpace}-completeness and \textsc{EXP}-completeness,
respectively. Our results leave open the possibility that the cost
problem for atomic cost formulas is not \textsc{EXP}-hard and even
in \textsc{PSpace}. The technique described in the proof sketch of
Thm.~\ref{thm-MDP-acyclic} cannot be applied to general cost
processes, because there we have to deal with paths of exponential
length, which, informally speaking, have double-exponentially small
probabilities. Proving hardness in an analogous way would thus require
probability thresholds~$\tau$ of exponential representation size.

\paragraph{Acknowledgements.} The authors would like to thank Andreas G\"obel for valuable hints, Christel Baier and Sascha Kl\"uppelholz for thoughtful feedback on an earlier version of this paper, and the anonymous referees for their helpful comments.

\bibliographystyle{plain}
\bibliography{references}

\iftechrep{
\newpage
\appendix

\addtolength{\abovecaptionskip}{0mm}
\addtolength{\belowcaptionskip}{0mm}
\addtolength{\topsep}{0mm}
\addtolength{\partopsep}{0mm}

\section{Proofs of Section~\ref{sec-prelim}} \label{app-prelim}

\newcommand{\stmtpropahalf}{
 Let $\C$ be a cost process, $\varphi$ a cost formula with $n_0 \not\models \varphi$ and $n_1 \models \varphi$ for some $n_0, n_1 \in \N$,
 and $\tau \in [0,1]$.
 One can construct in logarithmic space a cost process~$\C'$ such that the following holds:
 There is a scheduler~$\sigma$ for~$\C$ with $\P_\sigma(K_\C \models \varphi) \ge \tau$ if and only if
 there is a scheduler~$\sigma'$ for~$\C'$ with $\P_{\sigma'}(K_{\C'} \models \varphi) \ge 1/2$.
 Moreover, $\C'$ is a cost chain if $\C$ is.
}
\begin{proposition} \label{prop-a-half}
\stmtpropahalf
\end{proposition}
\begin{proof}
 Let $\tau < 1/2$.
 Define $p := (1/2 - \tau) / (1 - \tau)$.
 To construct~$\C'$ from~$\C$, add a new initial state $s_{00}$ with exactly one enabled action, say~$a$,
  and set $\Delta(s_{00},a)(t,n_1) = p$ and $\Delta(s_{00},a)(s_0,0) = 1-p$.
 In a straightforward sense any scheduler for~$\C$ can be viewed as a scheduler for~$\C'$ and vice versa.
 Thus for any scheduler~$\sigma$ we have $\P'_{\sigma}(K \models \varphi) = p + (1-p) \cdot \P_\sigma(K \models \varphi)$.
 The statement of the proposition now follows from a simple calculation.

 Now let $\tau > 1/2$.
 Define $p := 1/(2 \tau)$.
 In a similar way as before, add a new initial state $s_{00}$ with exactly one enabled action~$a$,
  and set $\Delta(s_{00},a)(t,n_0) = 1-p$ and $\Delta(s_{00},a)(s_0,0) = p$.
 Thus we have $\P'_{\sigma}(K \models \varphi) = p \cdot \P_\sigma(K \models \varphi)$,
  and the statement of the proposition follows. \qed
\end{proof}

\section{Proofs of Section~\ref{sec-quantile}} \label{app-quantile}

In this section we prove Prop.~\ref{prop-apriori-upper-bound} from the main text:

\begin{qproposition}{\ref{prop-apriori-upper-bound}}
\stmtpropaprioriupperbound
\end{qproposition}
\begin{proof}
Define $n := |Q|$.
If $\pmin = 1$, then by our assumption on the almost-sure reachability of~$t$,
 the state~$t$ will be reached within $n$ steps, and the statement of the proposition follows easily.
So we can assume $\pmin < 1$ for the rest of the proof. 

Let $j \in \N$ be the smallest integer with
 \begin{align*}
  j \ge n \cdot \left( \frac{-\ln (1-\tau)}{\pmin^n} + 1 \right)\,.
 \end{align*}
It follows:
 \begin{align}
  \left\lfloor \frac j n \right\rfloor
  & \ge \frac{-\ln (1-\tau)}{\pmin^n} \notag \\
  & \ge \frac{\ln (1-\tau)}{\ln(1 - \pmin^n)} && \text{(as $x \le - \ln (1-x)$ for $x < 1$)}
      \label{eq-apriori-upper-bound-j}
 \end{align}

For $i \in \N$ and $q \in Q$ and a scheduler~$\sigma$,
 define $p_i(q,\sigma)$ as the probability that,
  if starting in~$q$ and using the scheduler~$\sigma$, more than~$i$ steps are required
   to reach the target state~$t$.
Define $p_i := \max \{p_i(q,\sigma) : q \in Q, \ \sigma \text{ a scheduler}\}$.
By our assumption on the almost-sure reachability of~$t$,
 regardless of the scheduler, there is always a path to~$t$ of length at most~$n$.
This path has probability at least $\pmin^n$,
 so $p_n \le 1 - \pmin^n$.
If a path of length~$\ell \cdot n$ does not reach~$t$,
 then none of its $\ell$ consecutive blocks of length~$n$ reaches~$t$,
  so we have $p_{\ell \cdot n} \le p_n^\ell$.
Hence we have:
\begin{align}
 p_j & \le p_{\lfloor j/n \rfloor \cdot n} && \text{(as $p_i \ge p_{i+1}$ for all $i \in \N$)} \notag \\
     & \le p_n^{\lfloor j/n \rfloor}                 && \text{(as argued above)} \notag \\
     & \le (1 - \pmin^n)^{\lfloor j/n \rfloor}       && \text{(as argued above)} \notag \\
     &  =  \exp\big(\ln(1 - \pmin^n) \cdot \lfloor j/n \rfloor\big) \notag \\
     & \le 1 - \tau && \text{(by~\eqref{eq-apriori-upper-bound-j})} \label{eq-apriori-upper-bound-pj}
\end{align}
Denote by~$T$ the random variable that assigns to a run the ``time'' (i.e., the number of steps)
 to reach~$t$ from~$s_0$.
Then we have for all schedulers~$\sigma$:
\begin{align*}
 \P_\sigma(K \le B)
  & = \P_\sigma(K \le j \cdot \kmax) && \text{(by the definition of~$B$)} \\
  & \ge P_\sigma(T \le j) && \text{(each step costs at most $\kmax$)} \\
  &  =  1 - P_\sigma(T > j) \\
  & \ge 1 - p_j && \text{(by the definition of $T$ and $p_i$)} \\
  & \ge \tau && \text{(by~\eqref{eq-apriori-upper-bound-pj})\;,}
\end{align*}
as claimed. \qed
\end{proof}

\section{Proofs of Section~\ref{sec-MC}} \label{app-MC}

\subsection{Proof of Thm.~\ref{thm-MC-acyclic}}

In this section we prove Thm.~\ref{thm-MC-acyclic} from the main text:

\begin{qtheorem}{\ref{thm-MC-acyclic}}
\stmtthmMCacyclic
\end{qtheorem}

\begin{proof}
  First we prove membership in~PP.
  Recall from the main text that the class~PP can be
  defined as the class of languages~$L$ that have a probabilistic
  polynomial-time bounded Turing machine~$M_L$ such that for all
  words~$x$ one has $x \in L$ if and only if $M_L$ accepts~$x$ with
  probability at least~$1/2$, see~\cite{Gill77:PP} and note that PP is
  closed under complement~\cite{Gill77:PP}.
  By Prop.~\ref{prop-a-half} it suffices to consider an instance of
  the cost problem with $\tau = 1/2$.  The problem can be decided by a
  probabilistic Turing machine that simulates the cost chain as
  follows: The Turing machine keeps track of the control state and the
  cost, and branches according to the probabilities specified in the
  cost chain.  It accepts if and only if the accumulated cost
  satisfies~$\varphi$.  Note that the acyclicity of the cost chain
  guarantees the required polynomial time bound.
  This proof assumes that the probabilistic Turing machine has access
   to coins that are biased according to the probabilities in the cost chain.
  As we show in~\cite[Lem.~1]{HK15},
   this can indeed be assumed for probabilistic polynomial-time bounded
   Turing machines.

  As stated in the main text, the lower bound follows from combining
  the result from~\cite[Prop.~4]{LSdurational05} with~\cite[Thm.~3]{HK15}.
  \qed
\end{proof}

\subsection{Proof of Thm.~\ref{thm-MC-cyclic}}

In this section we prove Thm.~\ref{thm-MC-cyclic} from the main text:

\begin{qtheorem}{\ref{thm-MC-cyclic}}
\stmtthmMCcyclic
\end{qtheorem}

\medskip\noindent
First we give details on the upper bound.
Then we provide proofs of the statements from the main text that
pertain to the \textsc{PosSLP} lower bound.

\subsubsection{Proof of the Upper Bound in Thm.~\ref{thm-MC-cyclic}.} \label{app-MC-upper}
We show that the cost problem for cost chains is in \textsc{PSpace}.
As outlined in the main text we use the fact that \textsc{PSpace} equals probabilistic \textsc{PSpace}.
The cost problem for cost chains is in this class.
This can be shown in the same way as we showed in Thm.~\ref{thm-MC-acyclic}
 that the cost problem for acyclic cost chains is in PP.
More concretely, given an instance of the cost problem for cost chains,
 we construct in logarithmic space a probabilistic \textsc{PSpace} Turing machine
 that simulates the cost chain and accepts if and only if the accumulated cost~$K$ satisfies the given cost formula.

The fact that (this definition of) probabilistic \textsc{PSpace} equals \textsc{PSpace}
 was first proved in~\cite{Simon77}.
A simpler proof can be obtained using a result by Ladner~\cite{Ladner89}
 that states that \#\textsc{PSpace} equals~\textsc{FPSpace}, see~\cite{Ladner89} for definitions.
This was noted implicitly, e.g., in~\cite[Thm.~5.2]{MGLA00complexityMDP}.
We remark that the class \textsc{PPSpace} defined in~\cite{Papa83}
 also equals~\textsc{PSpace}, but its definition (which is in terms of stochastic games) is different.

\subsubsection{Proof of Prop.~\ref{lem:parikh-hardness}.} 

Here, we give a formal definition and proof of the construction
outlined in the main text which allows for computing the value of an
arithmetic circuit as the number of paths in a DFA with a certain
Parikh image. First, we formally define the notations informally used
in the main text.

We first introduce arithmetic circuits and at the same time take
advantage of a normal form that avoids gates labelled with~``$-$''.
This normal form was established in the proof
of~\cite[Thm. 5.2]{EY09rmc}.  An \emph{arithmetic circuit} is a
directed acyclic graph $G=(V,E)$ whose leaves are labelled with
constants ``$0$'' and ``$1$'', and whose vertices are labelled with
operators ``$+$'' and ``$\ast$''. Subsequently, we refer to the
elements of~$V$ as \emph{gates}.  With every gate we associate a level
starting at $0$ with leaves. For levels greater than zero, gates on
odd levels are labelled with~``$+$'' and on even levels
with~``$\ast$''.  Moreover, all gates on a level greater than zero
have exactly two incoming edges from the preceding level.  The upper
part of Fig.~\ref{fig-reduction} illustrates an arithmetic circuit in
this normal form. We can associate with every gate $v\in V$ a
non-negative integer $\mathit{val}(v)$ in the obvious way.  In this
form, the \textsc{PosSLP} problem asks, given an arithmetic circuit
$G=(V,E)$ and two gates $v_1, v_2 \in V$, whether $\mathit{val}(v_1)
\ge \mathit{val}(v_2)$ holds.

Regarding the relevant definitions of Parikh images, let
$\mathcal{A}=(Q,\Sigma,\Delta)$ be a DFA such that $Q$ is a finite set
of \emph{control states}, $\Sigma=\{a_1,\ldots,a_k\}$ is a
\emph{finite alphabet}, and $\Delta \subseteq Q\times \Sigma \times Q$
is the set of \emph{transitions}. A \emph{path} $\pi$ in $\mathcal{A}$
is a sequence of transitions $\pi=\delta_1\cdots \delta_n\in \Delta^*$
such that $\delta_i=(q_i,a_i,q_i')$ and
$\delta_{i+1}=(q_{i+1},a_{i+1},q_{i+1}')$ implies $q_i'=q_{i+1}$ for
all $1\le i<n$. Let $q,q'\in Q$, we denote by $\Pi(\mathcal{A},q,q')$
the set of all paths starting in $q$ and ending in $q'$. In this
paper, a \emph{Parikh function} is a function $f:\Sigma\to
\mathbb{N}$.  The \emph{Parikh image} of a path $\pi$, denoted
$\mathit{parikh}(\pi)$, is the unique Parikh function counting for
every $a\in \Sigma$ the number of times $a$ occurs on a transition in
$\pi$.

The following statement of Prop.~\ref{lem:parikh-hardness}
 makes the one given in the main text more precise.

\begin{qproposition}{\ref{lem:parikh-hardness}}
Let $G=(V,E)$ be an arithmetic
  circuit and $v\in V$. There
  exists a log-space computable DFA $\mathcal{A}=(Q,\Sigma,\Delta)$
  with distinguished control states $q,q'\in Q$ and a Parikh function $f:
  \Sigma \to \mathbb{N}$ such that
  \begin{align*}
    \val(v) = |\{ \pi\in \Pi(\mathcal{A},q,q') : \mathit{parikh}(\pi)=f \}|.
  \end{align*}
\end{qproposition}
\begin{proof}
  We construct $\mathcal{A}$ by induction on the number of levels of
  $V$. For every level~$i$, we define an alphabet $\Sigma_i$ and a
  Parikh function $f_i:\Sigma_i\to \mathbb{N}$. As an invariant,
  $\Sigma_i\subseteq \Sigma_{i+1}$ holds for all levels $i$. Subsequently,
  denote by $V(i)$ all gates on level $i$. For every $v\in V(i)$,
  we define a DFA $\mathcal{A}_v$ such that each $\mathcal{A}_v$ has
  two distinguished control locations $\mathit{in}(\mathcal{A}_v)$ and
  $\mathit{out}(\mathcal{A}_v)$. The construction is such that
  \begin{align}
    \label{eqn:invariant}
    \mathit{val}(v) =
    |\{ \pi\in \Pi(\mathcal{A}_v, \mathit{in}(\mathcal{A}_v),
    \mathit{out}(\mathcal{A}_v)) :
    \mathit{parikh}(\pi)=f_i \}|.
  \end{align}
  For technical convenience, we allow transitions to be labelled with
  subsets $S\subseteq \Sigma$ which simply translates into an
  arbitrary chain of transitions such that each $a\in S$ occurs
  exactly once along this chain. We now proceed with the details of
  the construction starting with gates on level $0$.

  With no loss of generality we may assume that there are two gates
  $v$ and $w$ on level $0$ labelled with $0$ and $1$,
  respectively.
  Let $\Sigma_0 = \{a\}$ for some letter~$a$.
  The DFA $\mathcal{A}_{v}$ and $\mathcal{A}_{w}$ over $\Sigma_0$ is
  defined as follows: $\mathcal{A}_{w}$ has a single transition
  connecting $\mathit{in}(\mathcal{A}_w)$ with
  $\mathit{out}(\mathcal{A}_w)$ labelled with~$a$, whereas
  $\mathcal{A}_{v}$ does not have this transition. Setting
  $f_0(a) = 1$, it is easily checked that (\ref{eqn:invariant})
  holds for those DFA.

  For level $i+1$,
  we define $\Sigma_{i+1}=\Sigma_i \uplus \{ a_v, b_v, c_v : v\in
  V(i+1) \}$. Let $v\in V(i+1)$ be a gate on level $i+1$ such that
  $v$ has incoming edges from $u$ and $w$. Let
  $\mathcal{A}_u=(Q_u,\Sigma_i,\Delta_u)$ and
  $\mathcal{A}_w=(Q_w,\Sigma_i,\Delta_w)$ be the DFA representing $u$
  and $w$.
  Let $Q_v$ be a set of fresh control states.
  We define $\mathcal{A}_{v}=(Q_v\cup Q_u \cup Q_w,
  \Sigma_{i+1}, \Delta_v\cup \Delta_u\cup \Delta_w)$.
  The particularities of
  the construction depend on the type of~$v$.

  If $i+1$ is odd, i.e., the gates on this level are labelled
  with ``+'', then apart from the control states
  $\mathit{in}(\mathcal{A}_{v})$ and $\mathit{out}(\mathcal{A}_{v})$,
  the set~$Q_v$ contains three additional control states $q,q_1,q_2$.
  Further we set
  $\Delta_v=\{\delta_1,\ldots, \delta_7 \}$ such that
  \begin{itemize}
  \item $\delta_1=(\mathit{in}(\mathcal{A}_{v}), S_v, q)$, where $S_v=\{
    a_w,b_w, c_w : w\in V(i+1), v\neq w \}$;
  \item $\delta_2=(q,a_v,q_1)$ and $\delta_3=(q,b_v,q_2)$;
  \item $\delta_4=(q_1,b_v,\mathit{in}(\mathcal{A}_u))$ and
    $\delta_5=(q_2,a_v,\mathit{in}(\mathcal{A}_w))$; and
  \item
    $\delta_6=(\mathit{out}(\mathcal{A}_u),c_v,\mathit{out}(\mathcal{A}))$
    and
    $\delta_7=(\mathit{out}(\mathcal{A}_w),c_v,\mathit{out}(\mathcal{A}))$.
  \end{itemize}
  Informally speaking, we simply branch at $q$ into $\mathcal{A}_u$
  and $\mathcal{A}_w$, and this in turn enforces that the number of
  paths in $\Pi(\mathcal{A}_{v},\mathit{in}(\mathcal{A}_{v}),
  \mathit{out}(\mathcal{A}_{v}))$ on which $a_v$ occurs once equals
  the sum of $\mathit{val}(u)$ and $\mathit{val}(w)$. The reason
  behind using \emph{both} $a_v$ and $b_v$ is that it ensures that the
  case $u=w$ is handled correctly. Setting $f_{i+1}(a)=1$ if $a\in
  \Sigma_{i+1}\setminus \Sigma_i$, and $f_{i+1}(a)=f_{i}(a)$
  otherwise, we consequently have that (\ref{eqn:invariant}) holds
  since
  \begin{align*}
    & \hspace{5mm} |\{ \pi\in \Pi(\mathcal{A}_{v}, \mathit{in}(\mathcal{A}_{v}),
    \mathit{out}(\mathcal{A}_{v})) : \mathit{parikh}(\pi) = f_{i+1} \}|\\
    & = |\{ \pi\in \Pi(\mathcal{A}_u, \mathit{in}(\mathcal{A}_u),
    \mathit{out}(\mathcal{A}_u)) : \mathit{parikh}(\pi) = f_{i} \}|
    + \mbox{}\\ & \qquad \mbox{} + |\{ \pi\in \Pi(\mathcal{A}_w, \mathit{in}(\mathcal{A}_w),
    \mathit{out}(\mathcal{A}_w)) : \mathit{parikh}(\pi) = f_{i} \}|\\
    & = \mathit{val}(u) + \mathit{val}(w)\\
    & = \mathit{val}(v).
  \end{align*}

  The case of $i+1$ being even can be handled analogously, but instead
  of using branching we use sequential composition in order to
  simulate the computation of a gate labelled with~``$\ast$''. Apart
  from the control states $\mathit{in}(\mathcal{A}_{v})$ and
  $\mathit{out}(\mathcal{A}_{v})$, the set~$Q_v$ contains an additional
  control state $q$.
  Further we set $\Delta_v=\{\delta_1,\ldots, \delta_4 \}$
  such that
  \begin{itemize}
  \item $\delta_1=(\mathit{in}(\mathcal{A}_{v}), S_v, q)$, where
    $S_v=\{ a_w,b_w, c_w : w\in V(i+1), v\neq w \}$;
  \item $\delta_2=(q,a_v,\mathit{in}(\mathcal{A}_u))$;
  \item
    $\delta_3=(\mathit{out}(\mathcal{A}_u),b_v,\mathit{in}(\mathcal{A}_w))$;
    and
  \item
    $\delta_4=(\mathit{out}(\mathcal{A}_w),c_v,\mathit{out}(\mathcal{A}_v))$.
  \end{itemize}
  A difference to the case where $i+1$ is odd is that via the
  definition of $f_{i+1}$ we have to allow for paths that can traverse
  \emph{both} $\mathcal{A}_u$ and $\mathcal{A}_w$. Consequently, we
  define $f_{i+1}(a)=1$ if $a\in \Sigma_{i+1}\setminus \Sigma_i$, and
  $f_{i+1}(a)=2 f_i(a)$ otherwise. Similarly as above,
  (\ref{eqn:invariant}) holds since
  \begin{align*}
    & \hspace{5mm}  |\{ \pi\in \Pi(\mathcal{A}_{v}, \mathit{in}(\mathcal{A}_{v}),
    \mathit{out}(\mathcal{A}_{v})) : \mathit{parikh}(\pi) = f_{i+1} \}| \\
    & = |\{ \pi\in \Pi(\mathcal{A}_u, \mathit{in}(\mathcal{A}_u),
    \mathit{out}(\mathcal{A}_u)) : \mathit{parikh}(\pi) = f_{i} \}|
    \cdot \mbox{}\\ & \qquad \mbox{} \cdot |\{ \pi\in \Pi(\mathcal{A}_w, \mathit{in}(\mathcal{A}_w),
    \mathit{out}(\mathcal{A}_w)) : \mathit{parikh}(\pi) = f_{i} \}| \\
    & = \mathit{val}(u) \cdot \mathit{val}(w)\\
    & = \mathit{val}(v).
  \end{align*}

  Due to the inductive nature of the construction, the cautious reader
  may on the first sight cast doubt that the computation of
  $\mathcal{A}_v$ and $f$ can be performed in logarithmic
  space. However, a closer look reveals that the graph underlying
  $\mathcal{A}_v$ has a simple structure and its list of edges can be
  constructed without prior knowledge of the DFA on lower
  levels. Likewise, even though $f$ contains numbers which are
  exponential in the number of levels of $G$, the structure of $f$ is
  simple and only contains numbers which are powers of two, and hence
  $f$ is computable in logarithmic space as well.\qed
\end{proof}

\newcommand{\T}{\mathcal{T}}%

\subsubsection{Proof of Prop.~\ref{prop-SLPtoMC}.} 

The following statement of Prop.~\ref{prop-SLPtoMC}
 makes the one given in the main text more precise.

\begin{qproposition}{\ref{prop-SLPtoMC}}
Let $G=(V,E)$ be an arithmetic circuit.
  Let $v\in V$ be a gate on level~$\ell$ with odd~$\ell$.  There exist
  a log-space computable cost process~$\C$ and $T \in \N$ with
  $\P(K_\C = T) = \mathit{val}(v)/m$, where $ m =
  \exp_2(2^{(\ell-1)/2+1}-1)\cdot \exp_d(2^{(\ell-1)/2+1}-3)\,.  $ 
\end{qproposition}

\medskip
For a clearer proof
structure we define an intermediate formalism between DFA and cost
chains. A \emph{typed cost chain} $\T = (Q, q_0, t, \Gamma, \Delta)$
is similar to a cost chain, but with costs (i.e., natural numbers)
replaced with functions $\Gamma \to \N$.  The intuition is that
instead of a single cost, a typed cost chain keeps track of several
types of cost, and each type is identified with a symbol
from~$\Gamma$.  More precisely, $Q$ is a finite set of control states,
$q_0 \in Q$ is the initial control state, $t$ is the target control
state, $\Gamma$ is a finite alphabet, and $\Delta : Q \to \dist(Q
\times \N^\Gamma)$ is a probabilistic transition function.

A typed cost chain~$\T$ induces a Markov chain in the same way as a
cost chain does, but the state space is $Q \times \N^\Gamma$ rather
than~$Q \times \N$.  Formally, $\T$ induces the Markov chain $\D_\T =
(Q \times \N^\Gamma, (q_0,\vec{0}), \delta)$, where by~$\vec{0}$ we
mean the function $c: \Gamma \to \N$ with $c(a) = 0$ for all $a \in
\Gamma$, and $\delta(q,c)(q',c') = \Delta(q)(q', c'-c)$ holds for all
$q,q' \in Q$ and $c,c' \in \N^\Gamma$, where by $c'-c$ we mean $c'':
\Gamma \to \N$ with $c''(a) = c'(a) - c(a)$ for all $a \in \Gamma$.
As before, we assume that the target control state~$t$ is almost
surely reached.  We write $K_\T$ for the (multi-dimensional) random
variable that assigns a run in~$\D_\T$ the typed cost $c: \Gamma \to
\N$ that is accumulated upon reaching~$t$.

\begin{lemma}\label{lem-typed-cost-chain-hardness}
  Let $G=(V,E)$ be an arithmetic circuit.
  Let $v\in V$ be a gate on level~$\ell$ with odd~$\ell$.
  Let $d=|V|+1$.
  There exist a log-space computable typed cost chain $\T = (Q, q_0, t,
  \Gamma, \Delta)$ and $c: \Gamma \to \mathbb{N}$ such that $\P(K_\T =
  c) = \val(v)/m$, where
  \begin{align*}
    m  = \exp_2(2^{(\ell-1)/2+1}-1)\cdot \exp_d(2^{(\ell-1)/2+1}-3)\,.
  \end{align*}
\end{lemma}
\begin{proof}
  With no loss of generality we may assume that the maximum level of
  $V$ is $\ell$ and that $v$ is the only gate on level $\ell$. The idea is
  to translate the DFA obtained in Prop.~\ref{lem:parikh-hardness}
  into a suitable typed cost chain. Subsequently, we refer to the
  terminology used in the proof of Prop.~\ref{lem:parikh-hardness}.

  Let $\mathcal{A}=(Q,\Sigma,\Delta)$ be the DFA,
  $q_0=\mathit{in}(\mathcal{A}_v)$, $t=\mathit{out}(\mathcal{A}_v)$,
  and $f:\Sigma \to \mathbb{N}$ be the Parikh function obtained
  from Prop.~\ref{lem:parikh-hardness}. We define $\Gamma=\Sigma
  \uplus \{e_j : 1 \le j \le d \}$ and alter $\mathcal{A}$ as
  follows:
  \begin{itemize}
  \item for the gate $w \in V$ on level~0 labelled with~0,
    we add an edge from $\mathit{in}(\mathcal{A}_w)$ to~$t$
    labelled with $e_1$; and
  \item for every $w\in V$ such that $w\neq v$, we add $k$ edges
    labelled with $e_1,\ldots, e_k$ from $\mathit{out}(\mathcal{A}_w)$
    to $t$, where $k$ is the difference between $d$ and the
    number of outgoing edges from $\mathit{out}(\mathcal{A}_w)$.
  \end{itemize}

  The DFA $\mathcal{A}'=(Q,\Gamma,\Delta')$ obtained from this
  construction has the property that $t$ can be reached from any
  control state, and that the number of outgoing edges from any
  $\mathit{out}(\mathcal{A}_w)$ for $w\neq v$ is uniform. Finally, we
  define $c:\Gamma \to \mathbb{N}$ such that $c$ coincides with
  $f$ for all $a\in \Sigma$ and $c(e_j)=0$ for all $1\le j\le
  k$. The intuition behind the $e_j$ is that they indicate errors, and
  once an edge with an $e_j$ is traversed it is impossible to reach
  $t$ with Parikh image~$c$. Thus,
  in particular property (\ref{eqn:invariant}) is preserved in
  $\mathcal{A}'$.

  We now transform $\mathcal{A}'$ into a typed cost chain $\T$.
  Subsequently, for $a\in \Gamma$ let $c_a:\Gamma\to \{0,1\}$ be the
  function such that $c_a(b)=1$ if $b=a$ and $c_a(b)=0$ otherwise. For
  our transformation, we perform the following steps:
  \begin{itemize}
  \item every alphabet letter $a\in \Gamma$ labelling a transition of
    $\mathcal{A}'$ is replaced by $c_a$;
  \item the probability distribution over edges is chosen uniformly;
    and
  \item a self-loop labelled with $\vec{0}$ and probability 1 is added
    at $t$.
  \end{itemize}
  We observe that the transition probabilities of $\T$ are either
  $1/d$, $1/2$ or $1$.
  Since $t$ can be reached from any control state, it is eventually reached with probability~$1$.

  For every level~$i$ and every $w \in V(i)$, let $p_w$ denote the probability that,
  starting from~$\mathit{in}(\mathcal{A}_w)$, the control state $\mathit{out}(\mathcal{A}_w)$ is reached
  \emph{and} typed cost~$c_i$ is accumulated.
  Here, $c_i$ refers to the Parikh function~$f_i$ constructed in the proof of
  Prop.~\ref{lem:parikh-hardness}, where we assert that $c_i(a)=0$ for
  all $a\in \Gamma$ on which the ``original'' $f_i$ is
  undefined.
  Since $t = \mathit{out}(\mathcal{A}_v)$ is almost surely reached
  from $q_0=\mathit{in}(\mathcal{A}_v)$, we have $p_v = \P(K_\T = c_\ell)$.
  So in order to prove the lemma, it suffices to prove for all $i \in \N$:
  \[
   p_w = \frac{\mathit{val}(w)}{m(i)} \qquad \text{for all $w \in V(i)$,}
  \]
  where
  \begin{align*}
    m(i) & = \left\{
    \begin{array}{ll}
      \exp_2(2^{i/2+1}-2)\cdot \exp_d(2^{i/2+1}-4) & \text{ if } i \text{ is even}\\
      \exp_2(2^{(i-1)/2+1}-1)\cdot \exp_d(2^{(i-1)/2+1}-3) & \text{ if } i
      \text{ is odd\,.}
    \end{array}\right.
  \end{align*}
  We proceed by induction on the level~$i$.
  Let $i=0$.
  If $w$ is labelled with $1$ then there is exactly one
  outgoing transition from~$\mathit{in}(\mathcal{A}_w)$,
   and this transition goes to~$\mathit{out}(\mathcal{A}_w)$ and incurs cost~$c_0$.
  So we have $p_w = 1$ as required.
  If $w$ is labelled with $0$, then
   the only outgoing transition from~$\mathit{in}(\mathcal{A}_w)$ incurs cost~$c$ with $c(e_1) = 1$,
   hence $p_w=0$.

  For the induction step, let $i \ge 0$.
  Let $w \in V(i+1)$ and let $u,u'\in V(i)$ be the gates connected to~$w$.
  If $i+1$ is odd then $w$ is labelled with ``+'', and by the
  construction of $\mathcal{A}_w$ and the transformation above we have
  \begin{align*}
    p_w & = \frac{1}{2} \cdot \frac{1}{d}\cdot ( p_u + p_{u'} ) \\
    & = \frac{1}{2} \cdot \frac{1}{d}\cdot \frac{\mathit{val}(u) + \mathit{val}(u')}
    {\exp_2(2^{i/2+1}-2)\cdot \exp_d(2^{i/2+1}-4)} & \text{by the ind.\ hypoth.}\\
    & = \frac{\mathit{val}(u) + \mathit{val}(u')}
    {\exp_2(2^{i/2+1}-1)\cdot \exp_d(2^{i/2+1}-3)} \\
    & = \ \frac{\mathit{val}(w)}{m(i+1)} .
  \end{align*}
  The factor $1/2$ is the probability of branching into
  $\mathit{in}(\mathcal{A}_u)$ or $\mathit{in}(\mathcal{A}_{u'})$, and
  $1/d$ is the probability that when leaving
  $\mathit{out}(\mathcal{A}_u)$ respectively
  $\mathit{out}(\mathcal{A}_{u'})$, the transition to
  $\mathit{out}(\mathcal{A}_w)$ is taken.

  Otherwise, if $i+1$ is even, we have
  \begin{align*}
    p_w & = \frac{1}{d^2} \cdot p_u \cdot p_{u'}
    \\
    & = \frac{1}{d^2}\cdot \frac{\mathit{val}(u) \cdot \mathit{val}
      (u')}
    {(\exp_2(2^{(i-1)/2+1}-1)\cdot \exp_d(2^{(i-1)/2+1}-3))^2} &
    \text{by the ind. hypoth.}\\
    & = \frac{1}{d^2} \cdot \frac{\mathit{val}(u) \cdot \mathit{val}(u')}
    {\exp_2(2^{(i-1)/2+2}-2)\cdot \exp_d(2^{(i-1)/2+2}-6)}\\
    & = \frac{\mathit{val}(u) \cdot \mathit{val}(u')}
    {\exp_2(2^{(i+1)/2+1}-2)\cdot \exp_d(2^{(i+1)/2+1}-4)} \\
    & = \frac{\mathit{val}(w)}{m(i+1)}.
  \end{align*}
  Here, $1/d^2$ is the probability that when leaving
  $\mathit{out}(\mathcal{A}_u)$ the transition to
  $\mathit{in}(\mathcal{A}_{u'})$ is taken, and that when leaving
  $\mathit{out}(\mathcal{A}_{u'})$ the transition to
  $\mathit{out}(\mathcal{A}_w)$ is taken. \qed
\end{proof}

In order to complete the proof of Prop.~\ref{prop-SLPtoMC}, we
now show how a typed cost chain can be transformed into a cost
chain. The idea underlying the construction in the next lemma is that
we can encode alphabet symbols into the digits of natural numbers
represented in a suitable base.

\begin{lemma}
  \label{lem:typed-costs-to-naturals}
  Let $\Gamma$ be a finite alphabet, and let $c, c_1, \ldots, c_n:
  \Gamma \to \N$ be functions represented as tuples with numbers
  encoded in binary. There exists a log-space computable homomorphism
  $h : \N^\Gamma \to \N$ such that for all $\lambda_1, \ldots,
  \lambda_n \in \N$ we have
  \begin{align*}
    \sum_{i=1}^n \lambda_i c_i = c \ \Longleftrightarrow \sum_{i=1}^n \lambda_i h(c_i) = h(c)\,.
  \end{align*}
\end{lemma}
  \begin{proof}
  Let $\Gamma=\{a_0,\ldots,a_{k-1}\}$, $m=\sum_{a\in \Gamma}c(a)$, and
  $b=m+1$. We define $h:\mathbb{N}^\Gamma \to \mathbb{N}$ as
  \begin{align*}
    h(d) & = d(a_0)\cdot b^0 + d(a_1)\cdot b^1 + \cdots + d(a_{k-1})\cdot b^{k-1} +
    \left(\sum\nolimits_{a\in \Gamma} d(a) \right) \cdot b^k.
  \end{align*}
  The homomorphism $h$ encodes any $d:\Gamma\to \mathbb{N}$ into the
  $k$ least significant digits of a natural number represented in base
  $b$, and the $k+1$-th digit serves as a check digit.

  Suppose $\sum_{i=1}^n \lambda_i c_i = c$. Then
  \begin{align*}
    \sum_{i=1}^n\lambda_i h(c_i) &
    = \sum\nolimits_{j=0}^{k-1}
    \left(\sum\nolimits_{i=1}^n\lambda_i c_i(a_j)\right)\cdot b^j +
     \left(\sum\nolimits_{a\in \Gamma}\sum\nolimits_{i=1}^n\lambda_i c_i(a)\right)
       \cdot b^k\\
       & = \sum\nolimits_{j=0}^{k-1} c(a_j)\cdot b^j + \sum\nolimits_{a\in \Gamma} c(a)
       \cdot b^k\\
       & = h(c).
  \end{align*}
  Conversely, assume that $\sum_{i=1}^n\lambda_ih(c_i)=h(c)$. By
  definition of $h$, the check digit $k+1$ ensures that
  \begin{align*}
    \sum_{a\in \Gamma}\sum_{i=1}^n\lambda_i c_i(a)
    = \sum_{a\in \Gamma} c(a) = m < b.
  \end{align*}
  Thus, in particular for a fixed $a_j\in \Sigma$ we have
  \begin{align*}
    \sum_{i=1}^n\lambda_i c_i(a_j) < b.
  \end{align*}
  But now, since $\sum_{i=1}^n\lambda_ih(c_i)=h(c)$ we have
  \begin{align*}
    \sum_{i=1}^n\lambda_i c_i(a_j) = c(a_j),
  \end{align*}
  and consequently $\sum_{i=1}^n\lambda_ic_i = c$. \qed
  \end{proof}
By replacing every typed cost function $c$ in $\T$ with $h(c)$, an
easy application of Lem.~\ref{lem:typed-costs-to-naturals} now yields
the following corollary.

\begin{corollary}
  Let $\T=(Q,q_0,t,\Gamma,\Delta)$ be a typed cost chain and
  $c:\Gamma\to \mathbb{N}$. There exist a log-space computable cost
  chain $\C=(Q,q_0,t,\delta)$ and $n\in \mathbb{N}$ such that
  \begin{align*}
    \P(K_\T = c) = \P(K_\C=n).
  \end{align*}
\end{corollary}
Together with Lem.~\ref{lem-typed-cost-chain-hardness}, this
completes the proof of Prop.~\ref{prop-SLPtoMC}.

\subsubsection{Proof of the Lower Bound in Thm.~\ref{thm-MC-cyclic}.}

Let $G=(V,E)$ be an arithmetic circuit with $v_1, v_2 \in V$. Without
loss of generality we assume that $v_1, v_2$ are on level~$\ell \in
\N$ with odd~$\ell$.  In the following, we construct in logarithmic
space a cost chain~$\C$ and a cost formula~$\varphi$ such that
\begin{equation}
  \label{eq-prop-SLPtoMCreduction}
  \mathit{val}(v_1) \ge \mathit{val}(v_2) \iff \P(K_\C \models
  \varphi) \ge 1/2\,.
\end{equation}

Using Prop.~\ref{prop-SLPtoMC} we first construct two cost chains
$\C_1 = (Q, q_1, t, \Delta)$ and $\C_2 = (Q, q_2, t, \Delta)$ and
$T_1,T_2 \in \N$ such that $\P(K_{\C_i} = T_i) = \mathit{val}(v_i)/m$
holds for $i \in \{1,2\}$ and for $m \in \N$ as given by
Prop.~\ref{prop-SLPtoMC}.  We compute a number $H \in \N$ with
$H \ge T_2$ such that $\P(K_{\C_2} > H) < 1/m$.  By
Prop.~\ref{prop-apriori-upper-bound}, it suffices to take
 \[
  H \ge  \max\left\{T_2, \ \kmax \cdot \left\lceil |Q| \cdot
   \left( \ln(m+1) \big/ \pmin^{|Q|} + 1 \right) \right\rceil\right\}\;,
 \]
  where $\kmax$ and~$\pmin$ refer to~$\C_2$.
Let $\varepsilon := \P(K_{\C_2} > H \ \land \ K_{\C_2} \ne H+1+T_1)$.
We have
\[
 0 \ \le \ \varepsilon \ \le \ \P(K_{\C_2} > H) \ < \ 1/m\,.
\]
Now we combine $\C_1$ and~$\C_2$ to a cost chain $\C = (Q \uplus \{q_0\}, q_0, t,
\widetilde\Delta)$, where $\widetilde\Delta$ extends~$\Delta$ by
\[
 \widetilde\Delta(q_0)(q_1,H+1) = 1/2 \quad \text{and} \quad \widetilde\Delta(q_0)(q_2,0) = 1/2\,.
\]
By this construction, the new cost chain~$\C$ initially either incurs
cost~$H+1$ and then emulates~$\C_1$, or incurs cost~$0$ and then
emulates~$\C_2$.  Those possibilities have probability $1/2$ each.  We
define the cost formula
\[
 \varphi := (x \le T_2-1) \ \lor \ (T_2+1 \le x \le H) \ \lor \ (x = H + 1 + T_1) \,.
\]
The construction of~$\C$ and the definition of~$\varepsilon$ gives that
$\P(K_\C \models \varphi)$ is equal to
\begin{align*}
 & \frac12 \cdot \P(K_{\C_1} = T_1) +
    \frac12 \cdot \big(\P(K_{\C_2} \le H \ \lor \ K_{\C_2} = H+1+T_1) - \P(K_{\C_2} = T_2) \big) \\
= & \frac12 \cdot \mathit{val}(v_1) / m + \frac12 \cdot \big(1 - \varepsilon - \mathit{val}(v_2) / m \big)
\end{align*}
It follows that we have $\P(K_\C \models \varphi) \ge 1/2$ if and only
if $\mathit{val}(v_1) / m \ge \mathit{val}(v_2) / m + \varepsilon$.
Since $0 \le \varepsilon < 1/m$ and
$\mathit{val}(v_1),\mathit{val}(v_2)$ are integer numbers, we have
shown the equivalence~\eqref{eq-prop-SLPtoMCreduction}.
This completes the proof of the \textsc{PosSLP} lower bound.
\qed

\medskip
Let us make two remarks on the construction just given:
First, the representation of~$m$ from
Prop.~\ref{prop-SLPtoMC} is of exponential size.  However, the
computation of~$H$ only requires an upper bound on the
\emph{logarithm} of $m+1$. Therefore, the reduction can be performed
in logarithmic space. Second, the structure of the cost
formula~$\varphi$, in particular the number of inequalities, is fixed.
Only the constants $T_1, T_2, H$ depend on the instance.


\section{Proofs of Section~\ref{sec-MDP}} \label{app-MDP}

\begin{qtheorem}{\ref{thm-MDP-acyclic}}
\stmtthmMDPacyclic
\end{qtheorem}

\begin{proof}
In the main body of the paper we proved the upper bound and gave a
sketch of the \textsc{PSpace}-hardness construction.  Following up on
this, we now provide the details of that reduction.

Let $\kmax := \max \{ k_1, k_2, \ldots, k_n\}$.
We choose $\ell := 1 + n \kmax$.
Before an action in control state~$q_{n{-}2}$ is played,
 at most the following cost is incurred:
\begin{align*}
\frac{n-2}{2} \cdot (\ell + 2 \kmax)
\ = \ \frac{n}{2} \cdot \ell + n \kmax - \ell - 2 \kmax
\ < \ \frac{n}{2} \cdot \ell
\ \le \ \frac{n}{2} \cdot \ell + T \ = \ B\;,
\end{align*}
so one cannot reach the full budget~$B$ before an action in control state~$q_{n{-}2}$ is played.
We choose
 \begin{equation}
  M := 2^{n/2} n^2 \ell^2  \quad \text{and}
   \quad \tau := \left(B - \frac12 \cdot \frac{1}{2^{n/2}}\right) \big/ M \;. \label{eq-choice-M}
 \end{equation}
For the sake of the argument we slightly change the standard \emph{old} MDP to a \emph{new} MDP,
 but without affecting the scheduler's winning chances.
The control state~$t$ is removed.
Any old transition~$\delta$ from~$q_i$ to~$t$ is redirected:
 with the probability of~$\delta$ the new MDP transitions to~$q_{i+2}$
  and incurs the cost of~$\delta$;
 in addition, one \emph{marble} is gained if the accumulated cost including the one of~$\delta$
  is at most the budget~$B$.
The idea is that a win in the old MDP (i.e., a transition to~$t$ having kept within budget)
 should correspond exactly to gaining \emph{at least one} marble in the new MDP.
The new MDP will be easier to analyse.

We make the definition of the new MDP more precise:
When in~$q_i$, the new MDP transitions to~$q_{i+2}$ with probability~$1$.
The cost incurred and marbles gained during that transition depend on the action taken and on probabilistic decisions as follows.
Suppose action~$a_j$ (with $j \in \{0,1\}$) is taken in~$q_i$, and cost~$C_i$ has been accumulated up to~$q_i$.
Then:
\begin{enumerate}
\item $j' \in \{0,1\}$ is chosen with probability $1/2$ each.
\item Cost $\ell + j \cdot k_{i+1} + j' \cdot k_{i+2}$ is incurred.
\item If $C_i + \ell + j \cdot k_{i+1} + j' \cdot k_{i+2} \le B$, then, in addition,
  one marble is gained with probability $\frac{\ell + j \cdot k_{i+1} + j' \cdot k_{i+2}}{M}$,
  and no marble is gained with probability $1 - \frac{\ell + j \cdot k_{i+1} + j' \cdot k_{i+2}}{M}$.
\end{enumerate}
In the new MDP, the scheduler's objective is, during the path from~$q_0$ to~$q_n$,
 to gain \emph{at least one marble}.
Since an optimal scheduler in the new MDP does not need to take into account
 whether or when marbles have been gained,
 we assume that schedulers in the new MDP do not take marbles into account.
The new MDP is constructed so that the winning chances are the same in the old and the new MDP;
 in fact, any scheduler in the old MDP translates into a scheduler with the same winning chance in the new MDP,
 and vice versa.

Fix a scheduler~$\sigma$ in the new MDP.
A vector $\vec{x} = (x_2, x_4, \ldots, x_n) \in \{0,1\}^{n/2}$ determines the cost incurred during a run,
 in the following way:
when $\sigma$ takes action~$a_j$ (for $j \in \{0,1\}$) in state~$q_i$,
 then cost
  $c_i(\sigma, \vec{x}) := \ell + j \cdot k_{i+1} + x_{i+2} \cdot k_{i+2}$
  is added upon transitioning to~$q_{i+2}$.
Conversely, a run determines the vector~$\vec{x}$.
Let $\hat{p}^\sigma_i(\vec{x})$ denote the conditional probability (conditioned under~$\vec{x}$)
 that a marble is gained upon transitioning from~$q_i$ to~$q_{i+2}$.
We have:
\begin{equation}
 \hat{p}^\sigma_i(\vec{x}) =
 \begin{cases} c_i(\sigma, \vec{x}) / M & \text{if }
  c_0(\sigma, \vec{x}) + c_2(\sigma, \vec{x}) + \cdots + c_i(\sigma, \vec{x}) \le B \\
               0 & \text{otherwise}
 \end{cases}
\label{eq-p-hat}
\end{equation}
It follows that we have:
\begin{align}
\hat{p}^\sigma_i(\vec{x}) & \ \le \ \frac{2 \ell}{M} \label{eq-bound-hat-p} \\
\sum_{\text{even }i=0}^{n-2} \hat{p}^\sigma_i(\vec{x}) & \ \le \ B/M \label{eq-total-p-upper}
\end{align}
Denote by $p^\sigma_i$ (for $i=0, 2, \ldots, n{-}2$) the (total) probability
 that a marble is gained upon transitioning from~$q_i$ to~$q_{i+2}$.
By the law of total probability we have
\begin{align}
 p^\sigma_i & = \sum_{\vec{x} \in \{0,1\}^{n/2}} \frac{1}{2^{n/2}} \cdot \hat{p}^\sigma_i(\vec{x}) \label{eq-total-probability} \\
 \text{and hence, by~\eqref{eq-bound-hat-p}, \qquad \ } p^\sigma_i & \le \frac{2 \ell}{M} \label{eq-bound-p}
\end{align}
We show that Player Odd has a winning strategy in the \textsc{QSubsetSum} game
 if and only if the probability of winning in the new MDP is at least~$\tau$.
\begin{itemize}
\item
Assume that Player Odd has a winning strategy in the \textsc{QSubsetSum} game.
Let $\sigma$ be the scheduler in the new MDP that emulates Player Odd's winning strategy from the \textsc{QSubsetSum} game.
Using~$\sigma$ the accumulated cost upon reaching~$q_n$ is exactly~$B$, with probability~$1$.
So for all $\vec{x} \in \{0,1\}^{n/2}$ we have:
\begin{align}
 c_0(\sigma, \vec{x}) + c_2(\sigma, \vec{x}) + \cdots + c_{n-2}(\sigma, \vec{x}) \ = \ B \label{eq-exactly-B}
\end{align}
Thus:
 \begin{equation}
 \begin{aligned}
  \sum_{\text{even }i=0}^{n-2} p^\sigma_i
  & \ = \ \sum_{\vec{x} \in \{0,1\}^{n/2}} \frac{1}{2^{n/2}}
        \sum_{\text{even }i=0}^{n-2} \hat{p}^\sigma_i(\vec{x}) && \text{\qquad by~\eqref{eq-total-probability}} \\
  & \ = \ \sum_{\vec{x} \in \{0,1\}^{n/2}} \frac{1}{2^{n/2}}
        \sum_{\text{even }i=0}^{n-2} c_i(\sigma, \vec{x}) / M
          && \text{\qquad by \eqref{eq-p-hat} and~\eqref{eq-exactly-B}} \\
  & \ = \ B/M && \text{\qquad by \eqref{eq-exactly-B}}
 \end{aligned}
 \label{eq-sum-pi}
 \end{equation}
Further we have:
\begin{equation} \label{eq-sum-pipj}
\begin{aligned}
\mathop{\sum_{\text{even }i,j}}_{i<j\le n{-}2} p^\sigma_i p^\sigma_j
& \ \mathop{\le}^\text{by~\eqref{eq-bound-p}} \ \binom{n/2}{2} \left( \frac{2 \ell}{M} \right)^2
  \ = \ \frac{\frac{n}{2} \cdot \left(\frac{n}{2} -1\right) \cdot 4 \ell^2}{2 M^2} \\
& \ \le \ \frac{n^2 \ell^2}{2 M^2} \ \mathop{=}^\text{by~\eqref{eq-choice-M}} \
  \left(\frac{1}{2} \cdot \frac{1}{2^{n/2}}\right) \big/ M
\end{aligned}
\end{equation}
Recall that the probability of winning equals the probability of gaining at least one marble.
The latter probability is, by the inclusion-exclusion principle, bounded below as follows:
 \[
  \sum_{\text{even }i=0}^{n-2} p^\sigma_i
      \ \ - \ \mathop{\sum_{\text{even }i,j}}_{i<j\le n{-}2} p^\sigma_i p^\sigma_j
      \ \mathop{\ge}^\text{by \eqref{eq-sum-pi} and~\eqref{eq-sum-pipj}} \
       \left(B - \frac12 \cdot \frac{1}{2^{n/2}}\right) \big/ M
      \ \mathop{=}^\text{by~\eqref{eq-choice-M}} \ \tau
 \]
We conclude that the probability of winning is at least~$\tau$.
\item
Assume that Player Odd does not have a winning strategy in the \textsc{QSubsetSum} game.
Consider any scheduler~$\sigma$ for the new MDP.
Since the corresponding strategy in the \textsc{QSubsetSum} game is not winning,
 there exists $\vec{y} \in \{0,1\}^{n/2}$ with
  $c_0(\sigma, \vec{y}) + c_2(\sigma, \vec{y}) + \cdots + c_{n-2}(\sigma, \vec{y}) \ne B$.
By~\eqref{eq-p-hat} it follows:
\begin{equation}
\sum_{\text{even }i=0}^{n-2} \hat{p}^\sigma_i(\vec{y}) \ \le \ (B-1)/M
\label{eq-total-p-upper-strict}
\end{equation}
By the union bound the probability of gaining at least one marble is bounded above as follows:
\[
 \begin{aligned}
  \sum_{\text{even }i=0}^{n-2} p^\sigma_i
  & \ = \ \sum_{\vec{x} \in \{0,1\}^{n/2}} \frac{1}{2^{n/2}}
        \sum_{\text{even }i=0}^{n-2} \hat{p}^\sigma_i(\vec{x}) && \text{\qquad by~\eqref{eq-total-probability}} \\
  & \ \le \ \left(1 - \frac{1}{2^{n/2}}\right) \cdot \frac{B}{M} \ + \ \frac{1}{2^{n/2}} \cdot \frac{B-1}{M}
  && \text{\qquad by \eqref{eq-total-p-upper} and~\eqref{eq-total-p-upper-strict}} \\
  & \ = \ \left(B - \frac{1}{2^{n/2}}\right) \big/ M \\
  & \ < \ \tau && \text{\qquad by~\eqref{eq-choice-M}}
 \end{aligned}
\]
We conclude that the probability of winning is less than~$\tau$.
\end{itemize}
This completes the log-space reduction. \qed
\end{proof}

\begin{qtheorem}{\ref{thm-MDP-general}}
\stmtthmMDPgeneral
\end{qtheorem}

\begin{proof}
We reduce from the problem of determining the winner in a countdown game~\cite{JurdzinskiSL08}.
A \emph{countdown game} is a tuple $(S, \mathord{\ctran{}}, s_0, T)$ where
 $S$ is a finite set of states,
 $\mathord{\ctran{}} \subseteq S \times \N\setminus\{0\} \times S$ is a transition relation,
 $s_0 \in S$ is the initial state,
 and $T$ is the final value.
We write $s \ctran{k} r$ if $(s, k, r) \in \mathord{\ctran{}}$.
A configuration of the game is an element $(s,c) \in S \times \N$.
The game starts in configuration $(s_0,0)$ and proceeds in moves:
 if the current configuration is $(s,c) \in S \times \N$,
  first Player~$1$ chooses a number~$k$ with $0 < k \le T - c$ and $s \ctran{k} r$ for at least one $r \in S$;
  then Player~$2$ chooses a state $r \in S$ with $s \ctran{k} r$.
The resulting new configuration is $(r,c+k)$.
Player~$1$ wins if she hits a configuration from $S \times \{T\}$, and she loses if she cannot move (and has not yet won).
(We have slightly paraphrased the game from~\cite{JurdzinskiSL08} for technical convenience,
 rendering the term \emph{countdown} game somewhat inept.)

The problem of determining the winner in a countdown game was shown
\textsc{EXP}-complete in~\cite{JurdzinskiSL08}.  Let $(S,
\mathord{\ctran{}}, s_0, T)$ be a countdown game.
We construct a cost process $\C = (Q, s_0, t, A, \En, \Delta)$
 so that Player~$1$ can win the countdown game if and only if there is a scheduler~$\sigma$
 with $\P_\sigma(K=T) = 1$.
The intuition is that Player~$1$ corresponds to the scheduler and Player~$2$ corresponds to randomness.
We take
\[
 Q := S \cup \{q_i : i \in \N, \ 2^i \le T\} \cup \{t\}\,. 
\]
Intuitively, the states in~$S$ are used in a \emph{first phase}, which directly reflects the countdown game.
The states $q_i$ are used in a \emph{second phase}, which is acyclic and ends in the final control state~$t$.

\newcommand{\astop}{a_\mathit{stop}}%
For all $s \in S$ we take
\[
 \En(s) := \{\astop\} \cup \{k \in \N \setminus \{0\} : \exists\, r \in S .\, s \ctran{k} r\}\,.
\]
Whenever $s \ctran{k} r$, we set $\Delta(s,k)(r,k) > 0$.
(We do not specify the exact values of positive probabilities, as they do not matter.
For concreteness one could take a uniform distribution.)
Those transitions directly reflect the countdown game.
Whenever $s \ctran{k} r$, we also set $\Delta(s,k)(q_0,k) > 0$.
Those transitions allow ``randomness'' to enter the second phase, which starts in~$q_0$.
Further, for all $s \in S$ we set $\Delta(s,\astop)(t,0) = 1$.
Those transitions allow the scheduler to jump directly to the final control state~$t$,
 skipping the second phase.

Now we describe the transitions in the second phase.
\newcommand{\imax}{{i_\textit{max}}}%
Let $\imax \in \N$ be the largest integer with $2^\imax \le T$.
For all $i \in \{0, 1, \ldots, \imax\}$ we take $\En(q_i) = \{a_0, a_1\}$ and
\begin{align*}
 \Delta(q_i,a_0)(q_{i+1},0)   & = 1 \text{\quad and} \\
 \Delta(q_i,a_1)(q_{i+1},2^i) & = 1 \;,
\end{align*}
where $q_{\imax+1}$ is identified with~$t$.
The second phase allows the scheduler to incur an arbitrary cost
 between $0$ and~$T$ (and possibly more).
That phase is acyclic and leads to~$t$.

Observe that $t$ is reached with probability~$1$.
We show that Player~$1$ can win the countdown game if and only if
 the scheduler in the cost process can achieve $K=T$ with probability~$1$.

Assume Player~$1$ can win the countdown game.
Then the scheduler can emulate Player~$1$'s winning strategy.
If randomness enters the second phase while the cost~$c$ accumulated so far is at most~$T$,
 then the scheduler incurs additional cost $T-c$ in the second phase and wins.
If and when accumulated cost exactly~$T$ is reached in the first phase,
 the scheduler plays~$\astop$, so it jumps to~$t$ and wins.
Since the scheduler emulates Player~$1$'s winning strategy,
 it will not get in a state in which the accumulated cost is larger than~$T$.

Conversely, assume Player~$2$ has a winning strategy in the countdown game.
If the scheduler jumps to~$t$ while accumulated cost~$T$ has not yet been reached, the scheduler loses.
If the scheduler does not do that,
 randomness emulates with non-zero probability Player~$2$'s winning strategy.
This leads to a state $(s,c) \in S \times \N$ with $c>T$,
 from which the scheduler loses with probability~$1$.
This completes the log-space reduction. \qed
\end{proof}

\begin{qcorollary}{\ref{cor-cost-utility}}
\stmtcorcostutility
\end{qcorollary}
\begin{proof}
Membership in \textsc{EXP} follows from Prop.~\ref{fact-EXPTIME-upper}.

For hardness, observe that
the proof of Thm.~\ref{thm-MDP-general} reveals that the following problem is
\textsc{EXP}-hard. The \emph{qualitative cost problem} asks, given a
cost process and $T \in \N$, whether there exists a scheduler~$\sigma$
with $\P_\sigma(K = T) = 1$.  
Reduce the qualitative
cost problem to the cost-utility problem where both the cost and the
utility in the new MDP are increased as the cost in the cost process.
Then we have $\P_\sigma(K = T) = 1$ in the cost process if and only if
in the new MDP the cost is at most~$T$ and the utility is at least~$T$
with probability~$1$. 
\qed
\end{proof}

\section{Proofs of Section~\ref{sec-extension}} \label{app-extension}

\begin{qtheorem}{\ref{thm-universal}}
\stmtthmuniversal
\end{qtheorem}
\begin{proof}
Considering the proof sketch in the main text,
it remains to show that the universal cost problem for acyclic cost processes and atomic cost formulas is \textsc{PSpace}-hard.
By a straightforward logspace reduction as in the beginning of the proof sketch in the main text, it suffices to prove \textsc{PSpace}-hardness of the following problem: given an acyclic cost process~$\C$ and a number $B \in \N$ and a probability~$\tau$, does there exist a scheduler~$\sigma$ with $\P_\sigma(K_\C < B) < \tau$?

For that we adapt the reduction from Theorem~\ref{thm-MDP-acyclic}.
The differences to that reduction arise from the fact that the scheduler now tries to maximise the probability of achieving cost \emph{at least}~$B$.
Given an instance $(k_1, \ldots, k_n, T)$, where $n$ is even, of the
\textsc{QSubsetSum} problem, we take, as before,
 \[
  B := \frac{n}{2} \cdot \ell + T \,,
 \]
for an $\ell \in \N$ defined later.
Further, we construct an acyclic cost process $\C = (Q, q_0, t, A, \En, \Delta)$ similarly as before.
In particular, we take again $Q = \{q_0, q_2, \ldots, q_{n-2}, q_n, t\}$.
For a large number $M \in \N$, defined later, we set for all even $i \le n{-}2$ and for $j \in \{0,1\}$:
\begin{align*}
 \Delta(q_i,a_j)(q_{i+2},\ell+j \cdot k_{i+1})           &= \frac12 \cdot \left(1 - \frac{\ell+j \cdot k_{i+1}}{M} \right) \\
 \Delta(q_i,a_j)(q_{i+2},\ell+j \cdot k_{i+1}+k_{i+2}) &= \frac12 \cdot \left(1 - \frac{\ell+j \cdot k_{i+1}+k_{i+2}}{M} \right) \\
 \Delta(q_i,a_j)(t,0)                                   &= \frac12 \cdot \frac{\ell+j \cdot k_{i+1}}{M} + \frac12 \cdot \frac{\ell+j \cdot k_{i+1}+k_{i+2}}{M}%
\intertext{%
So with a high probability the MDP transitions from~$q_i$ to~$q_{i+2}$,
 and cost $\ell$, $\ell+k_{i+1}$, $\ell+k_{i+2}$, $\ell+k_{i+1}+k_{i+2}$ is incurred,
 depending on the scheduler's (i.e., Player Odd's) actions and on the random (Player Even) outcome.
But with a small probability, which is proportional to the cost that would be otherwise incurred,
 the MDP takes a zero-cost transition to~$t$, which is a ``loss'' for the scheduler, because, as in the old reduction, $\ell$ is chosen big enough so that the total cost is strictly smaller than~$B$ before an action in control state~$q_{n-2}$ has been played.
There is a single zero-cost transition from $q_n$ to~$t$:
}
 \Delta(q_n,a)(t,0) &= 1 \qquad \text{with $\En(q_n) = \{a\}$}
\end{align*}
The MDP is designed so that the scheduler probably ``wins'' (i.e., reaches cost at least~$B$), if Player Odd can always reach cost at least~$T$;
but whenever cost~$k$ is incurred, there is a small probability~$k/M$ of losing.
Since $1/M$ is small, the overall probability of losing is approximately $C/M$ if total cost $C \ge B$ is incurred.
In order to minimise this probability, the scheduler wants to minimise the total cost while still incurring cost at least~$B$,
so the optimal scheduler will target~$B$ as total cost.

Similarly to the old reduction, the values for $\ell$, $M$ and~$\tau$ need to be chosen carefully, as the overall probability of losing is not exactly the sum of the individual losing probabilities.
Rather, this sum is -- by the ``union bound'' -- only an upper bound.
One needs to show that the sum approximates the real probability closely enough.

Now we give the details.
Let $\kmax := \max \{ k_1, k_2, \ldots, k_n\}$.
We choose $\ell := 1 + n \kmax$ as in the old reduction,
so one cannot reach the full budget~$B$ before an action in control state~$q_{n{-}2}$ is played.
Without loss of generality we can assume that $T \le n \kmax$, as otherwise the instance of the \textsc{QSubsetSum} problem would be trivial.
Hence we have:
\begin{equation}
 B+1 \ = \ \frac{n}{2} \ell + T + 1 \ \le \ \frac{n}{2} \ell + \ell \ \le \ n \ell \ \le \ n^2 \ell^2
 \label{eq-univ-B-upper-bound}
\end{equation}
We choose
 \begin{equation}
  M := 2^{n/2} n^2 \ell^2  \quad \text{and}
   \quad \tau := \left(B + \frac12 \cdot \frac{1}{2^{n/2}}\right) \big/ M \;. \label{eq-univ-choice-M}
 \end{equation}
For the sake of the argument we slightly change the standard \emph{old} MDP to a \emph{new} MDP,
 but without affecting the scheduler's winning chances.
The new MDP will be easier to analyse.
The control state~$t$ is removed.
When in~$q_i$, the new MDP transitions to~$q_{i+2}$ with probability~$1$.
The cost incurred and marbles gained during that transition depend on the action taken and on probabilistic decisions as follows.
Suppose action~$a_j$ (with $j \in \{0,1\}$) is taken in~$q_i$, and cost~$C_i$ has been accumulated up to~$q_i$.
Then:
\begin{enumerate}
\item $j' \in \{0,1\}$ is chosen with probability $1/2$ each.
\item Cost $\ell + j \cdot k_{i+1} + j' \cdot k_{i+2}$ is incurred.
\item 
  One marble is gained with probability $\frac{\ell + j \cdot k_{i+1} + j' \cdot k_{i+2}}{M}$,
  and no marble is gained with probability $1 - \frac{\ell + j \cdot k_{i+1} + j' \cdot k_{i+2}}{M}$.
\end{enumerate}
In the new MDP, the scheduler's objective is, during the path from~$q_0$ to~$q_n$,
 \emph{to gain no marble and to accumulate cost at least~$B$}.
Since an optimal scheduler in the new MDP does not need to take into account
 whether or when marbles have been gained,
 we assume that schedulers in the new MDP do not take marbles into account.
The new MDP is constructed so that the winning chances are the same in the old and the new MDP;
 in fact, any scheduler in the old MDP translates into a scheduler with the same winning chance in the new MDP,
 and vice versa.

Fix a scheduler~$\sigma$ in the new MDP.
A vector $\vec{x} = (x_2, x_4, \ldots, x_n) \in \{0,1\}^{n/2}$ determines the cost incurred during a run,
 in the following way:
when $\sigma$ takes action~$a_j$ (for $j \in \{0,1\}$) in state~$q_i$,
 then cost
  $c_i(\sigma, \vec{x}) := \ell + j \cdot k_{i+1} + x_{i+2} \cdot k_{i+2}$
  is added upon transitioning to~$q_{i+2}$.
Conversely, a run determines the vector~$\vec{x}$.
Let $\hat{p}^\sigma_i(\vec{x})$ denote the conditional probability (conditioned under~$\vec{x}$)
 that a marble is gained upon transitioning from~$q_i$ to~$q_{i+2}$.
We have:
\begin{align}
 \hat{p}^\sigma_i(\vec{x}) \ & = \ c_i(\sigma, \vec{x}) / M \label{eq-univ-p-hat} \\
 & \le \ 2 \ell / M \label{eq-univ-bound-hat-p}
\end{align}
Denote by $p^\sigma_i$ (for $i=0, 2, \ldots, n{-}2$) the (total) probability
 that a marble is gained upon transitioning from~$q_i$ to~$q_{i+2}$.
By the law of total probability we have
\begin{align}
 p^\sigma_i & = \sum_{\vec{x} \in \{0,1\}^{n/2}} \frac{1}{2^{n/2}} \cdot \hat{p}^\sigma_i(\vec{x}) \;. \nonumber \intertext{It follows:}
 p^\sigma_i & = \sum_{\vec{x} \in \{0,1\}^{n/2}} \frac{1}{2^{n/2}} \cdot c_i(\sigma, \vec{x}) / M 
 && \text{by~\eqref{eq-univ-p-hat}} \label{eq-univ-p-ext} \\
 p^\sigma_i & \le \frac{2 \ell}{M} 
 && \text{by~\eqref{eq-univ-bound-hat-p}} \label{eq-univ-bound-p}
\end{align}
We show that Player Odd has a winning strategy in the \textsc{QSubsetSum} game
 if and only if the probability of losing in the new MDP is less than~$\tau$.
\begin{itemize}
\item
Assume that Player Odd has a winning strategy in the \textsc{QSubsetSum} game.
Let $\sigma$ be the scheduler in the new MDP that emulates Player Odd's winning strategy from the \textsc{QSubsetSum} game.
Using~$\sigma$ the accumulated cost upon reaching~$q_n$ is exactly~$B$, with probability~$1$.
So for all $\vec{x} \in \{0,1\}^{n/2}$ we have:
\begin{align}
 c_0(\sigma, \vec{x}) + c_2(\sigma, \vec{x}) + \cdots + c_{n-2}(\sigma, \vec{x}) \ = \ B \label{eq-univ-exactly-B}
\end{align}
Since the scheduler accumulates, with probability~$1$, cost exactly~$B$, the probability of losing equals the probability of gaining at least one marble.
By the union bound this probability is bounded above as follows:
\[
 \begin{aligned}
  \sum_{\text{even }i=0}^{n-2} p^\sigma_i
  & \ = \ \sum_{\vec{x} \in \{0,1\}^{n/2}} \frac{1}{2^{n/2}}
        \sum_{\text{even }i=0}^{n-2} c_i(\sigma, \vec{x}) / M
          && \text{\qquad by \eqref{eq-univ-p-ext}} \\
  & \ = \ B/M && \text{\qquad by \eqref{eq-univ-exactly-B}} \\
  & \ < \ \tau && \text{\qquad by \eqref{eq-univ-choice-M}}
 \end{aligned}
\]
We conclude that the probability of losing is less than~$\tau$.
\item
Assume that Player Odd does not have a winning strategy in the \textsc{QSubsetSum} game.
Consider any scheduler~$\sigma$ for the new MDP.
Suppose that there exists $\vec{y} \in \{0,1\}^{n/2}$ with
  $c_0(\sigma, \vec{y}) + c_2(\sigma, \vec{y}) + \cdots + c_{n-2}(\sigma, \vec{y}) < B$.
Recall that the scheduler loses if it accumulates cost less than~$B$.
So the probability of losing is at least
\[
\frac{1}{2^{n/2}} \ = \ \frac{n^2 \ell^2}{2^{n/2} n^2 \ell^2} \ \mathop{\ge}^{\text{by~\eqref{eq-univ-B-upper-bound}}} \ \frac{B+1}{M} \ 
\mathop{\ge}^{\text{by~\eqref{eq-univ-choice-M}}} \ \tau\;.
\]
So we can assume for the rest of the proof that for all $\vec{x} \in \{0,1\}^{n/2}$ we have
\begin{equation}
 c_0(\sigma, \vec{x}) + c_2(\sigma, \vec{x}) + \cdots + c_{n-2}(\sigma, \vec{x}) \ \ge \ B\;.
\label{eq-univ-c-1}
\end{equation}
Since the strategy corresponding to~$\sigma$ in the \textsc{QSubsetSum} game is not winning,
 there exists $\vec{y} \in \{0,1\}^{n/2}$ with
\begin{equation}
  c_0(\sigma, \vec{y}) + c_2(\sigma, \vec{y}) + \cdots + c_{n-2}(\sigma, \vec{y}) \ \ge \ B+1\;.
\label{eq-univ-c-2}
\end{equation}
We have:
\begin{equation}
 \begin{aligned}
  & \qquad \sum_{\text{even }i=0}^{n-2} p^\sigma_i \\
  & \ = \ \sum_{\vec{x} \in \{0,1\}^{n/2}} \frac{1}{2^{n/2}}
        \sum_{\text{even }i=0}^{n-2} c_i(\sigma, \vec{x}) / M
          && \text{\qquad by \eqref{eq-univ-p-ext}} \\
  & \ \ge \ \left(1 - \frac{1}{2^{n/2}}\right) \cdot \frac{B}{M} \ + \ \frac{1}{2^{n/2}} \cdot \frac{B+1}{M}
  && \text{\qquad by \eqref{eq-univ-c-1} and~\eqref{eq-univ-c-2}} \\
  & \ = \ \left(B + \frac{1}{2^{n/2}}\right) \big/ M \\
 \end{aligned}
 \label{eq-univ-sum-pi}
\end{equation}
Further we have:
\begin{equation} \label{eq-univ-sum-pipj}
\begin{aligned}
\mathop{\sum_{\text{even }i,j}}_{i<j\le n{-}2} p^\sigma_i p^\sigma_j
& \ \mathop{\le}^\text{by~\eqref{eq-univ-bound-p}} \ \binom{n/2}{2} \left( \frac{2 \ell}{M} \right)^2
  \ = \ \frac{\frac{n}{2} \cdot \left(\frac{n}{2} -1\right) \cdot 4 \ell^2}{2 M^2} \\
& \ \le \ \frac{n^2 \ell^2}{2 M^2} \ \mathop{=}^\text{by~\eqref{eq-univ-choice-M}} \
  \left(\frac{1}{2} \cdot \frac{1}{2^{n/2}}\right) \big/ M
\end{aligned}
\end{equation}
Recall that the scheduler loses if it gains at least one marble.
So the probability of losing is, by the inclusion-exclusion principle, bounded below as follows:
 \[
  \sum_{\text{even }i=0}^{n-2} p^\sigma_i
      \ \ - \ \mathop{\sum_{\text{even }i,j}}_{i<j\le n{-}2} p^\sigma_i p^\sigma_j
      \ \mathop{\ge}^\text{by \eqref{eq-univ-sum-pi} and~\eqref{eq-univ-sum-pipj}} \
       \left(B + \frac12 \cdot \frac{1}{2^{n/2}}\right) \big/ M
      \ \mathop{=}^\text{by~\eqref{eq-univ-choice-M}} \ \tau
 \]
We conclude that the probability of losing is at least~$\tau$.
\end{itemize}
This completes the log-space reduction.
\qed
\end{proof}

}{}

\end{document}
